\documentclass[twocolumn]{autart}
\usepackage[round]{natbib}
\usepackage{graphicx}
\usepackage{graphpap}
\usepackage{amsmath}
\usepackage{amsfonts}
\usepackage{graphicx}
\usepackage{amssymb}
\usepackage{color}
\usepackage{ifpdf}
\usepackage{epstopdf}
\usepackage{algorithm, algpseudocode, algpascal}
\usepackage{array}
\usepackage{booktabs}
\usepackage{stfloats}
\usepackage{enumerate}
\usepackage{subfigure}

\definecolor{light}{gray}{.65}
\input{tcilatex}

\newtheorem{Rem}{Remark}
\begin{document}

\begin{frontmatter}

\title{Control theoretically explainable application of autoencoder methods to fault detection in nonlinear dynamic systems}

\thanks[footnoteinfo]{This work has been supported by the National Natural Science Foundation of
China under Grants 62073029 and 62173349.}

\author[linlin]{Linlin Li} \ead{linlin.li@ustb.edu.cn},  
\author[Ding]{Steven X. Ding}  \ead{steven.ding@uni-due.de}, 
\author[Ding]{Ketian Liang}  \ead{ketian.liang@uni-due.de}, 
\author[Chen]{Zhiwen Chen}  \ead{zhiwen.chen@csu.edu.cn}, and
\author[Xue]{Ting Xue}  \ead{xuetbuaa@126.com}

\address[linlin]{School of Automation and Electrical Engineering, University of Science and Technology Beijing, 100083 Beijing, China}
\address[Ding]{Institute for Automatic Control and Complex Systems, University of Duisburg-Essen, 47057 Duisburg, Germany}
\address[Chen]{Key Laboratory of Energy Saving Control and Safety
Monitoring for Rail Transportation of Hunan Provincial, School of
Automation, Central South University, Changsha 410083, China}
\address[Xue]{College of Electrical Engineering and Automation,
Shandong University of Science and Technology, Qingdao 266590, China}


\begin{abstract}                          
This paper is dedicated to control theoretically explainable application of
autoencoders to optimal fault detection in nonlinear dynamic systems.
Autoencoder-based learning is a standard machine learning method and widely applied for fault (anomaly) detection and classification. In the context of representation learning, the so-called latent (hidden) variable plays an important role towards an optimal fault detection. In ideal case, the latent variable should be a minimal sufficient statistic. The existing
autoencoder-based fault detection schemes are mainly application-oriented, and few efforts have been devoted to optimal autoencoder-based fault detection and explainable applications. The main objective of our work is to establish a framework for learning autoencoder-based optimal fault detection in nonlinear dynamic systems. To this aim, a process model form for dynamic systems is firstly introduced with the aid of control theory, which also leads to a clear system interpretation of the latent variable. The major efforts are made on the development of a control theoretic solution to the optimal fault detection problem, in which an analog concept to minimal sufficient statistic, the so-called lossless information compression, is introduced and proven for dynamic systems and fault detection specifications. In particular, the existence conditions for such a latent variable are derived, based on which a loss function and further a learning algorithm are developed. This learning algorithm enables optimally training of autoencoders to achieve an optimal fault detection in nonlinear dynamic systems. A case study on three-tank system is given at the end of this paper to illustrate the capability of the proposed autoencoder-based fault detection and to explain the essential role of the latent variable in the proposed fault detection system.

\end{abstract}

\end{frontmatter}

\section{Introduction}
Associated
with increasing demands on production efficiency
and system performance, today's industrial processes are of an extremely
high degree of complexity and nonlinearity. For such type of systems, safety
and reliability are of significant importance, which motivates the
development of fault detection methods \citep{Blanke06,Ding2008}. Reviewing
the publications on fault detection in nonlinear control systems shows that
the observer-based schemes serve as a major methodology \citep%
{TE_2002_FDI_SMO,ZPP_automatica_2010,YDL_SCL_2015,LDQY2017,LDQYX_2017}, and
promise reliable fault detection. The application of the observer-based
fault detection schemes requires a precise physical model of the system
under consideration, which demands for considerable modeling efforts and, in
turn, leads to high engineering costs. This calls for research endeavor to
develop data-driven fault detection approaches \citep{Ding2014}. Among the
involved studies, the subspace methods and multivariable statistic analysis
build the main research stream \citep%
{Huang_2008_book,YDXH_2014survey,Qin-survey2012,CDZLH-CCA-2016,LDYPQ2018}.
Nevertheless, these methods are incapable to handle highly nonlinear
dynamics and thus mainly limited to linear or linearized systems.

In recent years, machine learning (ML) based methods have drawn remarkably
increasing attention in both academic and industrial fields thanks to the
learning capacity and the ability in dealing with nonlinearities by means of
huge amount of process data \citep{Haykin2009}. One of the most popular
application areas of ML methods is classification, to which fault detection
(also known as anomaly detection), as a typical one-class classification
problem, belongs. Towards this end, intensive research efforts have been
made in representation learning \citep{BCV2013}. Roughly speaking, the basic
idea behind ML-based fault detection methods lies in the reconstruction of
process variables in the nominal process operation state. Different from the
model- and observer-based methods, the ML-based reconstruction of the
process variables under consideration is achieved by processing the process
data collected and recorded during fault-free process operations, known as
learning. Unfortunately, the learning capability often suffers considerably
from uncertainties in the collected process data, like disturbances or
variations caused by agings in assets. The so-called autoencoder (AE)
method, an efficient ML-based tool for dimensionality reduction and feature
learning, offers a reasonable and convincing solution to this problem \citep%
{Hinton2006,Goodfellow2016}. An autoencoder consists of an encoder that
compresses the process data into the so-called latent (hidden) variable, and
a decoder that is driven by the latent variable and reconstructs the process
variables reflecting nominal process operations \citep{SXY-TPAMI-2021}. Using
neural networks (NNs), the encoder and decoder are learnt, in optimal case,
in such a way that the latent variable contains exclusive informations
(features) of the nominal process operations, from which the process
variables in the nominal state are then fully reconstructed. The challenging
issue in this learning process is how to learn such an (ideal) latent
variable. In \citep{TZ-ITW-2015}, the concept of information bottleneck has
been introduced, which suggests that all the information from the input
variables required in the reconstructing (estimating) target should be
contained in the learned latent representation in the neural networks. In
the information theoretical framework, \citep{TZ-ITW-2015,Geiger2021} have
proposed to generate latent variable as the so-called minimal sufficient
statistic. On the basis of these concepts, the information plane analysis
has been popularized by Shwartz-Ziv and Tishby, which showcases the
important role of determining the latent representations aiming at high ML
performance \citep{RT-arxiv-2017}. The methods of information plane analysis
have been applied to construction of autoencoders as well \citep%
{YP-NN-2019,TE-NN-2020}.

A review of the existing AE-based fault diagnosis schemes gives the
impression that the main focus of the reported efforts has been on a direct
(and successful) application of the existing AE-algorithms and schemes to
dealing with fault diagnosis issues \citep%
{JIANG2017,Ahmed2021,HZP-CEP-2022,LYY-NCA-2021,RZZ-TII-2020}. In this
regard, the latent variable is generally viewed as features generated by the
learning process, without explainable interpretation and assessment of the
(generated) latent variable with respect to its information quality, e.g.
measured as minimal sufficient statistic. The consequence of such
application-oriented research efforts is that there is a lack of a
methodical framework for explainable applications of AE-based technique to
approach optimal fault detection issues.

A further noteworthy aspect is that few of the existing AE-based fault
detection schemes have been devoted to dynamic systems, as they are known in
control theory and exist widely in industry. Moreover, existing control and
system theoretic knowledge has been rarely utilized in those AE-based fault
detection schemes. It should be noticed that the dynamic systems considered
in our work differ from those typical objects and processes addressed by the
existing AE-based methods generally in

\begin{itemize}
\item their complex dynamics,

\noindent They are not only time evolutionary processes, but also driven by
process input variables, which are operation condition triggered, strongly
time-varying and often in feedback closed-loop configurations;

\item existence of hybrid uncertainties,

\noindent Typical uncertainties in industrial automatic control systems are
disturbances, including (external) unknown inputs, process and measurement
noises, variations in the environment around the process and within the
process e.g. caused by agings in the operational and control assets,
mismatching of embedded control loops, as well as errors generated during
data transmissions among the subsystems over networks.
\end{itemize}

\noindent Although most of the existing observer-based fault detection
methods successfully developed in the past decades \citep%
{TE_2002_FDI_SMO,ZPP_automatica_2010,YDL_SCL_2015,LDQY2017,LDQYX_2017}
cannot be, due to the lack of process models, applied to realize data-driven
fault detection, existing knowledge and ideas would be helpful to develop
capable AE-based methods to approach optimal fault detection in nonlinear
dynamic systems.

Motivated by the above discussions and observations, this paper is devoted
to the research effort of control theoretically guided application of
AE-based methods to approaching optimal fault detection in nonlinear dynamic
systems. The main objectives and the intended contributions are

\begin{itemize}
\item introduction of a process model form for dynamic systems;

\noindent This process model matches the configuration of an autoencoder
with a clear interpretation of the latent variable. To this aim, the
so-called coprime factorization technique will serve as a tool, and the
concepts of system image representation and subspace are introduced.

\item development of a control theoretic solution to optimal fault detection;

\noindent The centerpiece of this solution is the introduction of an analog
concept to minimal sufficient statistic for nonlinear dynamic systems and
learning of the latent variable. In particular, the existence conditions for
such a latent variable will be derived. Based on image representation and
subspace of nonlinear systems, methods of Hamiltonian extension and analysis
of inner systems will be applied.

\item information theoretic study on the proposed latent variable;

\noindent On assumption of a defined probabilistic setting and using the
concept of mutual information, it is proven that the proposed latent
variable is equivalent to a minimal sufficient statistic.

\item construction of an autoencoder;

\noindent With the aid of the control theoretic results as guidelines, an
autoencoder will be learnt that delivers a data-driven solution of the
optimal fault detection problem for nonlinear dynamic systems. The core of
this work is to recast the existence condition of the optimal latent
variable as regularized terms in the loss function for learning the
autoencoder.

\item realization of the autoencoder and test on an experimental system,
analysis of the fault detection performance of the developed autoencoder and
the role of the latent variable in approaching the optimal solution.
\end{itemize}

The paper is organized as follows. The preliminaries and problem formulation
are given in Section \ref{sub2}. Section \ref{sub3} includes the main
results and consists of four parts, (i) the basic ideas and optimal solution
illustrated by means of linear systems and the associated concepts, (ii) the
proof that the proposed latent variable is, in the context of mutual
information, equivalent to a minimal sufficient statistic, (iii) the optimal
solution for fault detection in nonlinear systems, and (iv) realization and
implementation of the optimal solution by means of an autoencoder. Finally,
in Section \ref{sub5}, the results on a case study on the laboratory setup
of a three-tank system are presented and analyzed.

\textbf{Notations}: Throughout this paper, standard notations known in
control theory, and in linear algebra are adopted. In addition, $T^{\sim
}(s)=T^{T}(-s)$ (respectively $T^{\sim }(z)=T^{T}(z^{-1}),$ $T^{\sim
}(e^{j\theta })=T^{T}(e^{-j\theta }):=T^{T}(-\theta )$) denotes the
conjugate of (rational) transfer function matrix $T(s)$ (respectively $T(z),$
$T(e^{j\theta })=T(z)\left\vert _{z=e^{j\theta }}\right. :=T(\theta )$), and 
$\mathcal{L}_{2}$ is the notation of the space of all square
summable/integrable Lebesgue signals (signals with bounded energy) \citep%
{Francis87,Vinnicombe-book}.

\section{Preliminaries and problem formulation\label{sub2}}

\subsection{Basics of data-driven fault detection paradigm for dynamic
systems}

Consider a nonlinear process $\Sigma :\mathcal{L}_{2}\rightarrow \mathcal{L}%
_{2}$, whose nominal dynamic is modelled by%
\begin{equation}
\Sigma :%
\begin{cases}
\dot{x}(t)=f(x(t),u(t)) \\ 
y(t)=c(x(t),u(t)).%
\end{cases}
\label{eq2-1a}
\end{equation}%
Here, $u\in \mathbb{R}^{p},y\in \mathbb{R}^{m},x\in \mathbb{R}^{n}$ denote
the process input, output and state vectors, respectively. $f(\cdot)$ and $%
c(\cdot)$ represent nonlinear continuous functions. Taking into account
possible disturbances in the process dynamic and measurement variables, the
above model is extended to 
\begin{equation}
\Sigma :%
\begin{cases}
\dot{x}(t)=\bar{f}(x(t),u(t),\eta (t)) \\ 
y(t)=\bar{c}(x(t),u(t),\varepsilon (t)),%
\end{cases}
\label{eq2-1b}
\end{equation}%
where $\eta ,\varepsilon $ represent unknown and $\mathcal{L}_{2}$-bounded
signals. $\bar{f}(\cdot)$ and $\bar{c}(\cdot)$ denote nonlinear continuous
functions. Model (\ref{eq2-1b}) represents the process dynamics during
fault-free operations.

As often met in industrial applications, it is supposed, in our subsequent
study, that

\begin{itemize}
\item the model (\ref{eq2-1b}) is unknown, and instead,

\item sufficient process data, $\left( u,y\right) ,$ are collected and
recorded during fault-free operations, and

\item they are available for the purpose of learnling process dynamics
during fault-free operations.
\end{itemize}

The major task of designing and operating a fault detection system is to
detect process operations that lead to a significant deviation of process
performance from its nominal value. It is a well-established paradigm of
approaching data-driven fault detection in dynamic systems with the
following steps and specifications:

\begin{itemize}
\item learn a dynamic system $\Pi $ using the collected fault-free process
data. $\Pi $ is driven by process data $\left( u,y\right) $ and delivers $%
\left( \hat{u},\hat{y}\right) $ serving as an estimate for the process input
and output variables in the nominal operations, i.e. 
\begin{equation}
\left[ 
\begin{array}{c}
\hat{u} \\ 
\hat{y}%
\end{array}
\right] :=\Pi \left( \left[ 
\begin{array}{c}
u \\ 
y%
\end{array}
\right] \right) ;  \label{eq2-3}
\end{equation}

\item learning $\Pi $ should satisfy, at a high probability, the
specifications that (i) for data $\left( u,y\right) $ generated during
fault-free operations, 
\begin{equation}
\left\Vert \left[ 
\begin{array}{c}
u \\ 
y%
\end{array}%
\right] -\left[ 
\begin{array}{c}
\hat{u} \\ 
\hat{y}%
\end{array}%
\right] \right\Vert ^{2}\leq J_{th},  \label{eq2-4}
\end{equation}%
hereby, signal vector 
\begin{equation*}
r:=\left[ 
\begin{array}{c}
u \\ 
y%
\end{array}%
\right] -\left[ 
\begin{array}{c}
\hat{u} \\ 
\hat{y}%
\end{array}%
\right] \in \mathbb{R}^{p+m}
\end{equation*}%
is often called residual (vector), (ii) in case of $\left( u,y\right) $
being generated by a faulty operation, 
\begin{equation*}
\left\Vert r\right\Vert ^{2}>J_{th},
\end{equation*}%
where $\left\Vert \cdot \right\Vert $ denotes a certain signal norm;

\item as a part of the learning process, the so-called threshold $J_{th}$
should be determined so that the following detection logic holds
\begin{equation}
\left\{ 
\begin{array}{l}
\left\Vert r\right\Vert ^{2}\leq J_{th},\text{ during fault-free operations}
\\ 
\left\Vert r\right\Vert ^{2}>J_{th},\text{ in case of faulty operations.}%
\end{array}%
\right.  \label{eq2-2a}
\end{equation}%
\end{itemize}

Fig. \ref{fig1A} schematically sketches the configuration and composition of
such a fault detection system. 
\begin{figure*}[h]
\centering\includegraphics[scale=0.42]{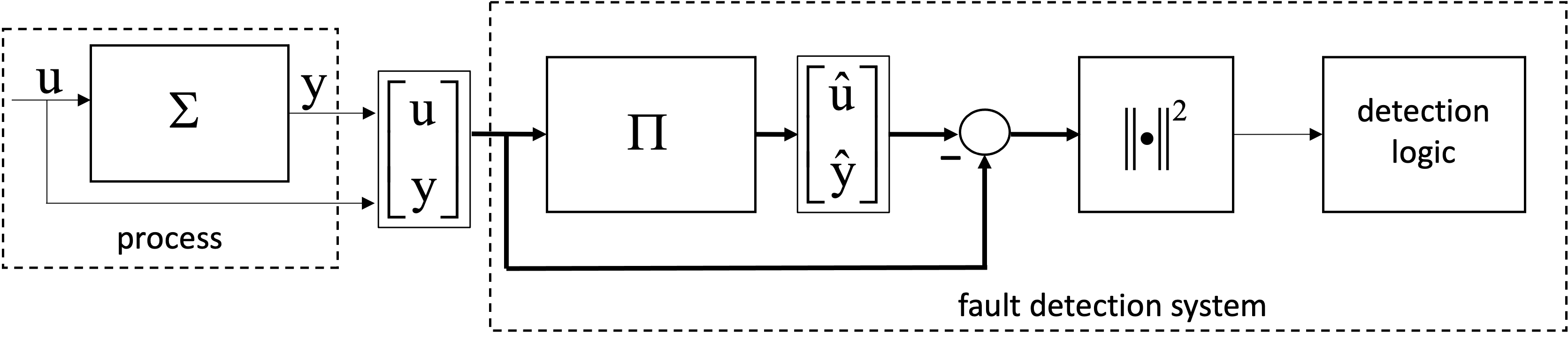}
\caption{Schematical description of a fault detection system. }
\label{fig1A}
\end{figure*}

It is noteworthy that the well-established observer-based fault detection
technique \citep{Ding2008} and its data-driven realization form \citep%
{Ding2014} are a special case of the above setting, in which $\hat{y}$ is
delivered by an observer and $\hat{u}=u.$ Consequently, computing%
\begin{equation*}
\left\Vert \left[ 
\begin{array}{c}
u \\ 
y%
\end{array}%
\right] -\left[ 
\begin{array}{c}
\hat{u} \\ 
\hat{y}%
\end{array}%
\right] \right\Vert \Longrightarrow \left\Vert \left[ 
\begin{array}{c}
u \\ 
y%
\end{array}%
\right] -\left[ 
\begin{array}{c}
u \\ 
\hat{y}%
\end{array}%
\right] \right\Vert =\left\Vert y-\hat{y}\right\Vert
\end{equation*}%
becomes an (output) residual generation and evaluation problem.

In engineering practice, the fault detection performance is mainly assessed
by the fault detectability subject to a user-defined upper bound of a false
alarm rate (FAR). Both in research and application domains, the fault
detectability and false alarm rate are often expressed by the probability of
successful detection of faulty operations and the probability of (false)
alarms in fault-free operations, respectively \cite{Ding2020}. In this
regard, data-driven design of an optimal fault detection system is
formulated as: given sufficient process data $\left( u,y\right) $ collected
during fault-free operations and the evaluation function $\left\Vert \cdot
\right\Vert ^{2}$, finding $\left( \Pi ,J_{th}\right) $ so that the fault
detectability is maximized while satisfying the user-defined FAR requirement.

\subsection{Basics of autoencoder technique and its applications to fault
detection}

Autoencoder methods are a well-established technique in ML \citep%
{Hinton2006,Goodfellow2016}. AE-based fault detection is one of numerous
application areas of AE methods and attracts increasing attention in recent
years \citep%
{JIANG2017,Ahmed2021,HZP-CEP-2022,LYY-NCA-2021,RZZ-TII-2020,YSXJ-KS-2021,TPD-ISA-2021,MYB-ANZCC-2020}%
. The basic idea behind the AE-based fault detection lies in the
reconstruction of process variables corresponding to nominal process
operations. To this aim, an autoencoder is learnt using process operation
data collected during fault-free operations. An AE is composed of two system
parts, an encoder that compresses the process variables into a
(low-dimensional) latent variable, and a decoder that reconstructs the
process variables from the latent variable. To be specific, consider a
process described by $\Sigma $ with process input and output variables $%
\left( u,y\right) .$ Using neural networks, $\mathcal{N}\mathcal{N}_{en}$
and $\mathcal{N}\mathcal{N}_{de}$ with $\theta _{en},\theta _{de}$ as the
associated parameters, the encoder and decoder are constructed as follows: 
\begin{gather}
v=\mathcal{N}\mathcal{N}_{en}\left( \theta _{en},\left[ 
\begin{array}{c}
u \\ 
y%
\end{array}%
\right] \right) ,\left[ 
\begin{array}{c}
\hat{u} \\ 
\hat{y}%
\end{array}%
\right] =\mathcal{N}\mathcal{N}_{de}\left( \theta _{de},v\right)
\Longrightarrow   \notag \\
\left[ 
\begin{array}{c}
\hat{u} \\ 
\hat{y}%
\end{array}%
\right] =\mathcal{N}\mathcal{N}_{de}\left( \theta _{de},\mathcal{N}\mathcal{N%
}_{en}\left( \theta _{en},\left[ 
\begin{array}{c}
u \\ 
y%
\end{array}%
\right] \right) \right) ,
\end{gather}%
where $v$ is the latent variable and $\left[ 
\begin{array}{c}
\hat{u} \\ 
\hat{y}%
\end{array}%
\right] $ denotes the reconstructed process variables that should reflect
the nominal process operations. Fig. \ref{fig2} showcases the structure of
an AE. 
\begin{figure}[h]
\centering\includegraphics[scale=0.49]{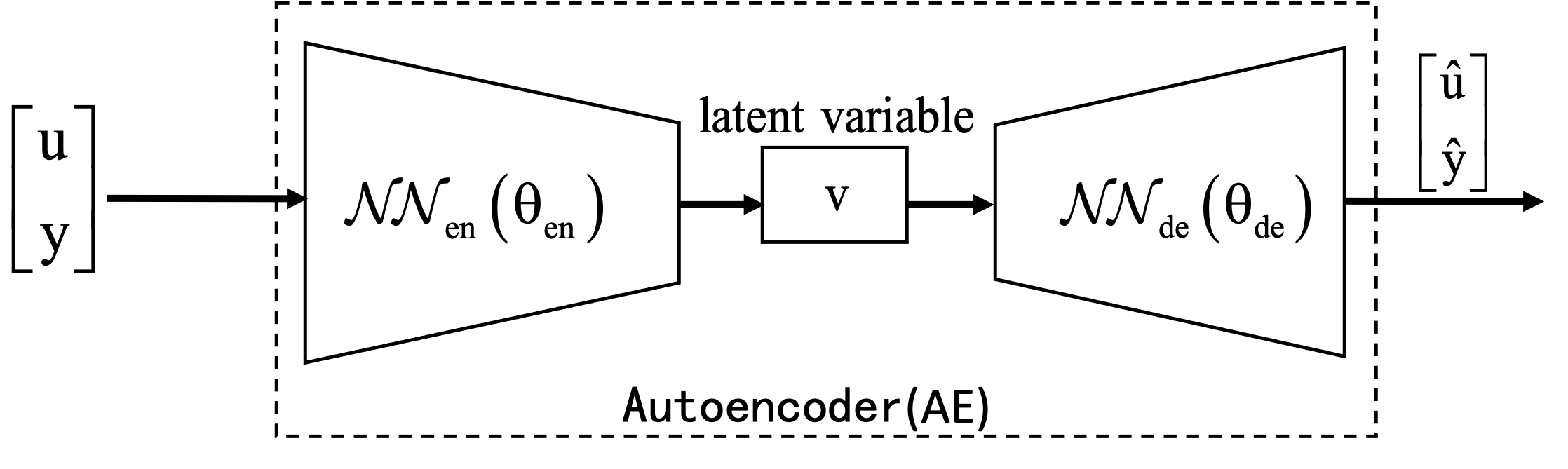}
\caption{The schematic of an autoencoder}
\label{fig2}
\end{figure}

To learn an AE, a standard and basic loss function is the squared
reconstruction error, for instance, defined by 
\begin{align}
& \mathcal{L}(\theta _{de},\theta _{en})\hspace{-2pt}=\hspace{-2pt}\frac{1}{M%
}\sum\limits_{i=1}^{M}\hspace{-2pt}\left( \hspace{-2pt}\left[ \hspace{-2pt}
\begin{array}{c}
u^{(i)} \\ 
y^{(i)}%
\end{array}%
\hspace{-2pt}\right] \hspace{-3pt}-\hspace{-3pt}\left[ \hspace{-2pt}
\begin{array}{c}
\hat{u}^{(i)} \\ 
\hat{y}^{(i)}%
\end{array}%
\hspace{-2pt}\right]\hspace{-2pt} \right) ^{T}\hspace{-4pt}\left(\hspace{-2pt} \left[\hspace{-2pt} 
\begin{array}{c}
u^{(i)} \\ 
y^{(i)}%
\end{array}%
\hspace{-2pt}\right] \hspace{-3pt}-\hspace{-3pt}\left[ \hspace{-2pt}
\begin{array}{c}
\hat{u}^{(i)} \\ 
\hat{y}^{(i)}%
\end{array}%
\hspace{-2pt}\right] \hspace{-2pt}\right) \hspace{-2pt},  \notag \\
& \left[ 
\begin{array}{c}
\hat{u}^{(i)} \\ 
\hat{y}^{(i)}%
\end{array}%
\right] =\mathcal{N}\mathcal{N}_{de}\left( \theta _{de},\mathcal{N}\mathcal{N%
}_{en}\left( \theta _{en},\left[ 
\begin{array}{c}
u^{(i)} \\ 
y^{(i)}%
\end{array}%
\right] \right) \right) .  \notag
\end{align}%
Here, $\left[ 
\begin{array}{c}
u^{(i)} \\ 
y^{(i)}%
\end{array}%
\right] $ represents a data sample (or batch) from the data set, and $M$
denotes the number of the data samples (batches) used for the learning
purpose. The parameters $\theta _{en},\theta _{de}$ are to be determined by
solving the optimization problem%
\begin{equation*}
\min_{\theta _{en},\theta _{de}}\mathcal{L}(\theta _{de},\theta _{en}).
\end{equation*}%
For the fault detection purpose, it is natural to use the trained
autoencoder (i) to generate 
\begin{equation*}
\left[ 
\begin{array}{c}
\hat{u} \\ 
\hat{y}%
\end{array}%
\right] =\mathcal{N}\mathcal{N}_{de}\left( \theta _{de},\mathcal{N}\mathcal{N%
}_{en}\left( \theta _{en},\left[ 
\begin{array}{c}
u \\ 
y%
\end{array}%
\right] \right) \right) 
\end{equation*}%
using online data $\left[ 
\begin{array}{c}
u \\ 
y%
\end{array}%
\right] ,$ (ii) to run the evaluation function 
\begin{equation*}
J=\left( \left[ 
\begin{array}{c}
u \\ 
y%
\end{array}%
\right] -\left[ 
\begin{array}{c}
\hat{u} \\ 
\hat{y}%
\end{array}%
\right] \right) ^{T}\left( \left[ 
\begin{array}{c}
u \\ 
y%
\end{array}%
\right] -\left[ 
\begin{array}{c}
\hat{u} \\ 
\hat{y}%
\end{array}%
\right] \right) ,
\end{equation*}%
and finally (iii) to make a decision according to the detection logic 
\begin{equation*}
\begin{cases}
J>J_{th}\Longrightarrow \text{faulty} \\ 
J\leq J_{th}\Longrightarrow \text{fault-free,}%
\end{cases}%
\end{equation*}%
where the threshold $J_{th}$ is set by means of the (minimum) value of the
loss function delivered by training.

\subsection{Problem formulation}

It is a widely recognized and accepted fact that, thanks to the power of
neural networks (including deep NNs) of approximating nonlinear functions
and systems, an autoencoder delivers optimal reconstruction of process
variables with respect to the defined loss function. For instance, when a
static (and statistic) process is under consideration, it has been proven
that an autoencoder with the squared reconstruction error as the loss
function is equivalent to the well-known principle component analysis (PCA)
algorithm, which is widely applied in fault detection and process monitoring
as well. Recently, efforts have been reported on improving fault detection
performance by introducing regularized terms into the loss function to
regularize the latent variable \cite%
{YSXJ-KS-2021,TPD-ISA-2021,MYB-ANZCC-2020}.

Comparing the optimal fault detection problem formulated in Subsection 2.1
and the basic principle of the AE-based fault detection technique introduced
in Subsection 2.2 leads to a convincing conclusion that the AE technique
offers an efficient tool to solve the formulated optimal fault detection
problem. On the other hand, to our best knowledge, no research work and
results have been reported on such a solution. To approach the optimal
solution, a challenging issue is how and to which degree the nominal process
operations can be fully reconstructed by means of the process data,
collected during fault-free operations but corrupted with noises or
operation uncertainties. In fact, this is the major concerning of
observer-based fault detection technique as well, which, as a special form
of the fault detection system (\ref{eq2-3})-(\ref{eq2-4}), is based on the
reconstruction of the process output variable using a nominal process model
and thus considerably suffers from uncertainties. In the ML framework, this
problem is reflected in a different form. Due to the capability of neural
networks to approximate nonlinear functions, overfitting often leads to
lower fault detectability, since uncertain operations would be learnt as a
part of the system dynamics. This issue has been addressed in the context of
the so-called information bottleneck \citep{BCV2013,Geiger2021}, which is
expressed in form of latent variables. An ideal latent variable should be
generated by maximally compressed mapping of the input variable that
preserves as much as possible the information on the output variable,
according to \cite{TZ-ITW-2015}. In the information theoretic framework, 
\cite{TZ-ITW-2015,Geiger2021} have proposed the concept minimal sufficient
statistic, and suggested to optimize autoencoders by generating the latent
variable in the sense of minimal sufficient statistic.

Motivated and inspired by the aforementioned discussions, the following
problems and tasks are formulated for our subsequent study:

\begin{itemize}
\item establishing a framework for the reconstruction of nominal process
variables $\left( u,y\right) $ on the basis of a latent variable. For this
purpose, the so-called system coprime factorization technique, that is
widely applied in observer-based fault diagnosis technique \citep{Ding2020},
will serve as a tool;

\item developing a control theoretic concept analog to minimal sufficient
statistic, which allows a maximal compression of information in the process
data $\left( u,y\right) $ about process nominal operations in terms of the
latent variable and leads to a reconstruction of the nominal operations
without loss of information. Control theoretic and mathematical existence
conditions for such a latent variable should be found;

\item studying on a probabilistic interpretation of the latent variable as a
minimal sufficient statistic with the aid of the established mutual
information concept;

\item constructing and learning autoencoders guided by the derived existence
conditions, which should result in an AE-based solution of the formulated
optimal fault detection problem, and finally,

\item verifying the proposed AE-based optimal fault detection system and, in
comparison with the standard AE-based schemes, analyzing the fault detection
capability.
\end{itemize}

\section{Main results\label{sub3}}

This section is devoted to a control theoretically guided learning of
autoencoders aiming at optimally detecting faults in nonlinear dynamic
systems. To this end, we first introduce and illustrate the basic idea in
the well-established framework of linear system theory as well as its
information theoretic interpretation. It is followed by a study on the
extension of the basic idea to nonlinear dynamic systems, and finally its
autoencoder-based realization.

\subsection{Introduction of the basic idea}

Consider a linear time invariant (LTI) system modelled by a transfer
function matrix $G(s),$ whose minimal state space representation is given by%
\begin{equation*}
G:%
\begin{cases}
\dot{x}(t)=Ax(t)+Bu(t) \\ 
y(t)=Cx(t)+Du(t),%
\end{cases}%
\end{equation*}%
where $A,B,C,D$ are system matrices of appropriate dimensions. A right
coprime factorization (RCF) of $G$ is given by $G(s)=N(s)M^{-1}(s)$ with the
right coprime pair $\left( M(s),N(s)\right) .$ The RCF of $G$ can be
interpreted as a state feedback control system with the closed-loop dynamics%
\begin{equation}
\begin{cases}
\dot{x}(t)=\left( A+BF\right) x(t)+BVv(t) \\ 
u(t)=Fx(t)+Vv(t) \\ 
y(t)=\left( C+DF\right) x(t)+DVv(t),%
\end{cases}
\label{eq3-2}
\end{equation}%
where $F$ as the state feedback gain is selected such that $A+BF$ is
Hurwitz, $V$ as a pre-filter is an invertible constant matrix, and $v$
serves as a reference signal. Correspondingly, $M(s),N(s)$ are stable
systems with the state space representations 
\begin{align}
& M(s)=(A+BF,BV,F,V),  \label{eq3-2a} \\
& N(s)=(A+BF,BV,C+DF,DV).  \label{eq3-2b}
\end{align}%
The system (\ref{eq3-2}) is called stable image representation (SIR) of $G$
and expressed in the frequency domain as \cite{Ding2020,Schaft-2000-book} 
\begin{equation}
I_{G}(s):\left[ 
\begin{array}{c}
u(s) \\ 
y(s)%
\end{array}%
\right] =\left[ 
\begin{array}{c}
M(s) \\ 
N(s)%
\end{array}%
\right] v(s).  \label{eq3-2c}
\end{equation}%
Note that the SIR of $G$ implies the input-output dynamic $y=Gu$, i.e. 
\begin{align*}
y(s)& =N(s)v(s),v(s)=M^{-1}(s)u(s) \\
& \Longrightarrow y(s)=N(s)M^{-1}(s)u(s)=G(s)u(s).
\end{align*}%
In this context, vector $v$ acts as a latent variable.

\begin{Rem}
Hereafter, we may drop out the domain variable $s$\ or $t$ when there is no
risk of confusion.
\end{Rem}

During normal process operations, the process data $\left( u,y\right) $
build a subspace in the Hilbert space $\mathcal{L}_{2},$ the so-called image
subspace of $G,$ which is explicitly defined by the SIR of $G$ and the
latent variable $v$ as follows 
\begin{equation}
\mathcal{I}_{G}=\left\{ \left[ 
\begin{array}{c}
u \\ 
y%
\end{array}%
\right] :\left[ 
\begin{array}{c}
u \\ 
y%
\end{array}%
\right] =\left[ 
\begin{array}{c}
M \\ 
N%
\end{array}%
\right] v,v\in \mathcal{L}_{2}\right\} .  \label{eq3-3}
\end{equation}%
It is of interest to notice that the nominal process data $\left( u,y\right) 
$ can be viewed as being generated by the latent variable $v$ serving as an
information signal (like the reference signal in a feedback control loop).

Now, we are in a position to introduce the basic idea of approaching the
problem of designing optimal fault detection systems formulated in the
previous section.

Let $\mathcal{P}_{\mathcal{V}}$ be an operator defined on a subspace $%
\mathcal{V}$ in Hilbert space that is endowed with the inner product, 
\begin{equation*}
\left\langle \alpha ,\beta \right\rangle =\int\limits_{-\infty }^{\infty
}\alpha ^{T}(t)\beta (t)dt,\alpha ,\beta \in \mathcal{V\subset L}_{2}.
\end{equation*}%
If $\mathcal{P}_{\mathcal{V}}$ is idempotent and self-adjoint, namely 
\begin{equation}
\forall \alpha ,\beta \in \mathcal{V},\mathcal{P}_{\mathcal{V}}^{2}=\mathcal{%
\ P}_{\mathcal{V}},\left\langle \mathcal{P}_{\mathcal{V}}\alpha ,\beta
\right\rangle =\left\langle \alpha ,\mathcal{P}_{\mathcal{V}}\beta
\right\rangle ,  \label{eq3-4}
\end{equation}%
it is an operator of an orthogonal projection onto $\mathcal{V}$ \cite%
{Kato_book}. The following properties of an orthogonal projection are of
importance for our solution \citep{Kato_book}:

\begin{itemize}
\item $\forall \alpha \in \mathcal{V},\mathcal{P}_{\mathcal{V}}\alpha
=\alpha \Longleftrightarrow \left\langle \mathcal{P}_{\mathcal{V}}\alpha
,\alpha -\mathcal{P}_{\mathcal{V}}\alpha \right\rangle =0;$

\item $\forall \alpha \in \mathcal{L}_{2},\alpha =\mathcal{P}_{\mathcal{V}%
}\alpha +\mathcal{P}_{\mathcal{V}^{\bot }}\alpha ,$ where $\mathcal{V}^{\bot
}$ is the orthogonal complement of $\mathcal{V};$

\item given $\beta \in \mathcal{L}_{2},\forall x\in \mathcal{V\subset L}%
_{2}, $ 
\begin{equation*}
\left\langle \beta -\alpha ,\beta -\alpha \right\rangle =\left\Vert \beta
-\alpha \right\Vert _{2}\geq \left\Vert \beta -\mathcal{P}_{\mathcal{V}%
}\beta \right\Vert _{2}.
\end{equation*}
\end{itemize}

\noindent Moreover, if the subspace $\mathcal{V}$ is closed, the distance
between $\beta $ and $\mathcal{V},dist\left( \beta ,\mathcal{V}\right) ,$ is
defined as%
\begin{equation}
dist\left( \beta ,\mathcal{V}\right) =\inf_{\alpha \in \mathcal{V}%
}\left\Vert \beta -\alpha \right\Vert _{2},  \label{eq3-5}
\end{equation}%
which can be computed as 
\begin{equation*}
dist\left( \beta ,\mathcal{V}\right) =\left\Vert \beta -\mathcal{P}_{%
\mathcal{V}}\beta\right\Vert _{2} =\left\Vert \mathcal{P}_{\mathcal{V}^{\bot
}}\beta \right\Vert _{2}.
\end{equation*}%
It is well-known that the image subspace $\mathcal{I}_{G}$ is closed in $%
\mathcal{L}_{2}$ and 
\begin{equation*}
\Pi (s)=I_{G,0}(s)I_{G,0}^{\sim }(s)
\end{equation*}%
forms an orthogonal projection onto $\mathcal{I}_{G}$ \cite{Vinnicombe-book}%
, denoted by $\mathcal{P}_{\mathcal{I}_{G}},$ namely%
\begin{equation}
\mathcal{P}_{\mathcal{I}_{G}}\left[ 
\begin{array}{c}
u \\ 
y%
\end{array}%
\right] =\Pi \left[ 
\begin{array}{c}
u \\ 
y%
\end{array}%
\right] =I_{G,0}I_{G,0}^{\sim }\left[ 
\begin{array}{c}
u \\ 
y%
\end{array}%
\right] .  \label{eq3-5a}
\end{equation}%
Here, $I_{G,0}$ is the normalized SIR of $G$ and satisfies 
\begin{equation}
I_{G,0}^{\sim }(s)I_{G,0}(s)=M_{0}^{\sim }(s)M_{0}(s)+N_{0}^{\sim
}(s)N_{0}(s)=I.  \label{eq3-5b}
\end{equation}%
The pair $\left( M_{0},N_{0}\right) $ is a RCF of $G$ with the following
setting for $F,V$ given in (\ref{eq3-2a})-(\ref{eq3-2b}) \cite{Hoffmann1996}%
, 
\begin{align*}
F& =-\left( I+D^{T}D\right) ^{-1}\left( D^{T}C+B^{T}P\right) , \\
V& =\left( I+D^{T}D\right) ^{-1/2},
\end{align*}%
where $P>0$ is the solution to the following Riccati equation%
\begin{align*}
& A^{T}P\hspace{-2pt}+\hspace{-2pt}PA\hspace{-2pt}+\hspace{-2pt}C^{T}C%
\hspace{-2pt}-\hspace{-2pt}\left( D^{T}C\hspace{-2pt}+\hspace{-2pt}%
B^{T}P\right) ^{T}R^{-1}\left( D^{T}C+B^{T}P\right) \hspace{-2pt}=\hspace{%
-2pt}0, \\
& R\hspace{-2pt}=\hspace{-2pt}I\hspace{-2pt}+\hspace{-2pt}D^{T}D.
\end{align*}%
Note that operator $\mathcal{I-P}_{\mathcal{I}_{G}}:\mathcal{L}%
_{2}\rightarrow \mathcal{L}_{2},$ 
\begin{equation*}
\left( \mathcal{I-P}_{\mathcal{I}_{G}}\right) \left[ 
\begin{array}{c}
u \\ 
y%
\end{array}%
\right] =\left( I-I_{G,0}I_{G,0}^{\sim }\right) \left[ 
\begin{array}{c}
u \\ 
y%
\end{array}%
\right] ,
\end{equation*}%
defines an orthogonal projection onto the orthogonal complement of $\mathcal{%
I}_{G},$ denoted by $\mathcal{I}_{G}^{\bot }.$ Consequently, any process
data can be written as 
\begin{equation*}
\left[ 
\begin{array}{c}
u \\ 
y%
\end{array}%
\right] =\mathcal{P}_{\mathcal{I}_{G}}\left[ 
\begin{array}{c}
u \\ 
y%
\end{array}%
\right] +\mathcal{P}_{\mathcal{I}_{G}^{\bot }}\left[ 
\begin{array}{c}
u \\ 
y%
\end{array}%
\right] ,\mathcal{P}_{\mathcal{I}_{G}^{\bot }}=\mathcal{I-P}_{\mathcal{I}%
_{G}}.
\end{equation*}%
In the context of one-class classification, faulty operations are detected
if 
\begin{equation*}
\left[ 
\begin{array}{c}
u \\ 
y%
\end{array}%
\right] \notin \mathcal{I}_{G}\Longrightarrow \mathcal{P}_{\mathcal{I}%
_{G}^{\bot }}\left[ 
\begin{array}{c}
u \\ 
y%
\end{array}%
\right] \neq 0,
\end{equation*}%
and $\mathcal{P}_{\mathcal{I}_{G}^{\bot }}\left[ 
\begin{array}{c}
u \\ 
y%
\end{array}%
\right] $ is sufficiently large (with respect to a defined threshold, see
below). In other words, in order to achieve a reliable and optimal fault
detection, the test statistic or the residual evaluation function should be
maximally sensitive to $\mathcal{P}_{\mathcal{I}_{G}^{\bot }}\left[ 
\begin{array}{c}
u \\ 
y%
\end{array}%
\right] $. In order to gain a deeper insight into the concepts concerning $%
\mathcal{I}_{G},\mathcal{I}_{G}^{\bot }$ and normalized coprime
factorizations, which are useful in our subsequent work, we briefly
introduce some essential relations.

As the dual concepts to RCF and SIR, the so-called left coprime
factorisation (LCF) and stable kernel representation (SKR) of $G$\ are
well-established in factorization technique \citep%
{Vidyasagar85,Vinnicombe-book}. Denoted by 
\begin{equation*}
K_{G}(s)=\left[ 
\begin{array}{cc}
-\hat{N}(s) & \text{ }\hat{M}(s)%
\end{array}%
\right]
\end{equation*}%
with $\left( \hat{M},\hat{N}\right) $ as a left coprime pair, the SKR $K_{G}$
satisfies 
\begin{equation}
\left[ 
\begin{array}{cc}
-\hat{N}(s) & \text{ }\hat{M}(s)%
\end{array}%
\right] \left[ 
\begin{array}{c}
M(s) \\ 
N(s)%
\end{array}%
\right] =0.  \label{eq3-31}
\end{equation}%
By means of $K_{G}$, the kernel subspace of $G$ is defined as%
\begin{equation}
\mathcal{K}_{G}=\left\{ \left[ 
\begin{array}{c}
u \\ 
y%
\end{array}%
\right] :\left[ 
\begin{array}{cc}
-\hat{N} & \text{ }\hat{M}%
\end{array}%
\right] \left[ 
\begin{array}{c}
u \\ 
y%
\end{array}%
\right] =0,\left[ 
\begin{array}{c}
u \\ 
y%
\end{array}%
\right] \in \mathcal{L}_{2}\right\} ,
\end{equation}%
which is, due to relation (\ref{eq3-31}), identical with $\mathcal{I}_{G},$
i.e. $\mathcal{K}_{G}=\mathcal{I}_{G}$ \citep{Vinnicombe-book}.
Correspondingly, $\mathcal{I}_{G}^{\bot }$ can be defined by 
\begin{equation*}
\mathcal{I}_{G}^{\bot }=\left\{ \left[ 
\begin{array}{c}
u \\ 
y%
\end{array}%
\right] :\left[ 
\begin{array}{cc}
-\hat{N} & \text{ }\hat{M}%
\end{array}%
\right] \left[ 
\begin{array}{c}
u \\ 
y%
\end{array}%
\right] \neq 0,\left[ 
\begin{array}{c}
u \\ 
y%
\end{array}%
\right] \in \mathcal{L}_{2}\right\} .
\end{equation*}%
Let $\left( \hat{M}_{0},\hat{N}_{0}\right) $ be the normalized left coprime
pair that, as a dual form of $\left( M_{0},N_{0}\right) ,$ satisfies 
\begin{equation}
\hat{N}_{0}(s)\hat{N}_{0}^{\sim }(s)+\hat{M}_{0}(s)\hat{M}_{0}^{\sim }(s)=I.
\label{eq3-32}
\end{equation}%
On account of (\ref{eq3-5b}), (\ref{eq3-31}) and (\ref{eq3-32}), we have%
\begin{equation}
\left[ 
\begin{array}{cc}
M_{0}^{\sim } & \text{ }N_{0}^{\sim } \\ 
-\hat{N}_{0} & \text{ }\hat{M}_{0}%
\end{array}%
\right] \left[ 
\begin{array}{cc}
M_{0} & \text{ }-\hat{N}_{0}^{\sim } \\ 
N_{0} & \text{ }\hat{M}_{0}^{\sim }%
\end{array}%
\right] =\left[ 
\begin{array}{cc}
I & \text{ }0 \\ 
0 & \text{ }I%
\end{array}%
\right] .  \label{eq3-33}
\end{equation}%
It follows from (\ref{eq3-33}) that%
\begin{align*}
\left[ 
\begin{array}{c}
u \\ 
y%
\end{array}%
\right] & =\left[ 
\begin{array}{cc}
M_{0} & \text{ }-\hat{N}_{0}^{\sim } \\ 
N_{0} & \text{ }\hat{M}_{0}^{\sim }%
\end{array}%
\right] \left[ 
\begin{array}{c}
v \\ 
q%
\end{array}%
\right] =\left[ 
\begin{array}{c}
M_{0} \\ 
N_{0}%
\end{array}%
\right] v+\left[ 
\begin{array}{c}
-\hat{N}_{0}^{\sim } \\ 
\hat{M}_{0}^{\sim }%
\end{array}%
\right] q, \\
\left[ 
\begin{array}{c}
v \\ 
q%
\end{array}%
\right] & =\left[ 
\begin{array}{cc}
M_{0}^{\sim } & \text{ }N_{0}^{\sim } \\ 
-\hat{N}_{0} & \text{ }\hat{M}_{0}%
\end{array}%
\right] \left[ 
\begin{array}{c}
u \\ 
y%
\end{array}%
\right] =\left[ 
\begin{array}{c}
M_{0}^{\sim }u+N_{0}^{\sim }y \\ 
\hat{M}_{0}y-\hat{N}_{0}u%
\end{array}%
\right] ,
\end{align*}%
and furthermore%
\begin{equation}
\mathcal{I}_{G}^{\bot }=\left\{ \left[ 
\begin{array}{c}
u \\ 
y%
\end{array}%
\right] :\left[ 
\begin{array}{c}
u \\ 
y%
\end{array}%
\right] =\left[ 
\begin{array}{c}
-\hat{N}_{0}^{\sim } \\ 
\hat{M}_{0}^{\sim }%
\end{array}%
\right] q,q\in \mathcal{L}_{2}\right\} .  \label{eq3-34}
\end{equation}

\begin{Rem}
It is noteworthy that the above results on coprime factorizations,
orthogonal projections, and image and kernel subspaces hold both for
continuous-time and discrete-time systems. The reader is referred to \cite%
{Vidyasagar85,Hoffmann1996,Vinnicombe-book,ding2022} for more details about
the aforementioned methods.
\end{Rem}

Now, we delineate how to solve the formulated optimal fault detection
problem by means of the system 
\begin{equation*}
\left[ 
\begin{array}{c}
\hat{u} \\ 
\hat{y}%
\end{array}%
\right] =\Pi \left[ 
\begin{array}{c}
u \\ 
y%
\end{array}%
\right] =I_{G,0}I_{G,0}^{\sim }\left[ 
\begin{array}{c}
u \\ 
y%
\end{array}%
\right] ,
\end{equation*}%
step by step and on account of the following arguments:

\begin{itemize}
\item given process data $\left( u,y\right) $ generated by normal
operations, there exists $v\in \mathcal{L}_{2}$ so that 
\begin{gather*}
\left[ 
\begin{array}{c}
\hat{u} \\ 
\hat{y}%
\end{array}%
\right] =I_{G,0}I_{G,0}^{\sim }I_{G,0}v=I_{G,0}v\in \mathcal{I}%
_{G}\Longrightarrow \\
\left\Vert \left[ 
\begin{array}{c}
u \\ 
y%
\end{array}%
\right] -\left[ 
\begin{array}{c}
\hat{u} \\ 
\hat{y}%
\end{array}%
\right] \right\Vert _{2}=0;
\end{gather*}

\item given process data $\left( u,y\right) $ generated during operations
with disturbances/uncertainties or faults, 
\begin{gather}
\left[ 
\begin{array}{c}
u \\ 
y%
\end{array}%
\right] -\left[ 
\begin{array}{c}
\hat{u} \\ 
\hat{y}%
\end{array}%
\right] =\left( I-I_{G,0}I_{G,0}^{\sim }\right) \left[ 
\begin{array}{c}
u \\ 
y%
\end{array}%
\right]  \notag \\
\Longrightarrow \left\Vert \left[ 
\begin{array}{c}
u \\ 
y%
\end{array}%
\right] -\left[ 
\begin{array}{c}
\hat{u} \\ 
\hat{y}%
\end{array}%
\right] \right\Vert _{2}=\left\Vert \mathcal{P}_{\mathcal{I}_{G}^{\bot }}%
\left[ 
\begin{array}{c}
u \\ 
y%
\end{array}%
\right] \right\Vert _{2}\neq 0;\,  \label{eq3-7c}
\end{gather}

\item according to the distance definition (\ref{eq3-5}), the threshold is
set to be 
\begin{align}
J_{th}& :=\sup_{\substack{ \left[ 
\begin{array}{c}
u \\ 
y%
\end{array}%
\right] \in \mathcal{Z}  \\ FAR\leq \gamma }}dist\left( \left[ 
\begin{array}{c}
u \\ 
y%
\end{array}%
\right] ,\mathcal{I}_{G}\right) ^{2}  \notag \\
& =\sup_{\substack{ \left[ 
\begin{array}{c}
u \\ 
y%
\end{array}%
\right] \in \mathcal{Z}  \\ FAR\leq \gamma }}\inf_{\left[ 
\begin{array}{c}
u_{0} \\ 
y_{0}%
\end{array}%
\right] \in \mathcal{I}_{G}}\left\Vert \left[ 
\begin{array}{c}
u \\ 
y%
\end{array}%
\right] -\left[ 
\begin{array}{c}
u_{0} \\ 
y_{0}%
\end{array}%
\right] \right\Vert _{2}^{2},  \label{eq3-8}
\end{align}%
where $\mathcal{Z}$ denotes the set of the process data $\left( u,y\right) $
collected during the fault-free operations, and $\gamma $ is the
user-defined FAR upper bound.
\end{itemize}

\noindent The following facts are of remarkable importance in the subsequent
study:

\begin{itemize}
\item for any $\left[ 
\begin{array}{c}
u \\ 
y%
\end{array}
\right] \in \mathcal{L}_{2},$ 
\begin{align}
\left[ 
\begin{array}{c}
\hat{u} \\ 
\hat{y}%
\end{array}
\right] & =\mathcal{P}_{\mathcal{I}_{G}}\left[ 
\begin{array}{c}
u \\ 
y%
\end{array}
\right] \in \mathcal{I}_{G}\Longleftrightarrow  \notag \\
\left[ 
\begin{array}{c}
\hat{u} \\ 
\hat{y}%
\end{array}
\right] & =I_{G,0}v,v=I_{G,0}^{\sim }\left[ 
\begin{array}{c}
u \\ 
y%
\end{array}
\right] \in \mathcal{L}_{2},  \label{eq3-6}
\end{align}

\item moreover, 
\begin{equation}
\left\Vert \left[ 
\begin{array}{c}
\hat{u} \\ 
\hat{y}%
\end{array}
\right] \right\Vert _{2}=\left\Vert \mathcal{P}_{\mathcal{I}_{G}}\left[ 
\begin{array}{c}
u \\ 
y%
\end{array}
\right] \right\Vert _{2}=\left\Vert I_{G,0}^{\sim }\left[ 
\begin{array}{c}
u \\ 
y%
\end{array}
\right] \right\Vert _{2}=\left\Vert v\right\Vert _{2}  \label{eq3-7a}
\end{equation}

\item as well as 
\begin{align}
\left[ 
\begin{array}{c}
u \\ 
y%
\end{array}
\right] & =\mathcal{P}_{\mathcal{I}_{G}}\left[ 
\begin{array}{c}
u \\ 
y%
\end{array}
\right] +\mathcal{P}_{\mathcal{I}_{G}^{\bot }}\left[ 
\begin{array}{c}
u \\ 
y%
\end{array}
\right] \Longrightarrow  \notag \\
\left\Vert \left[ 
\begin{array}{c}
u \\ 
y%
\end{array}
\right] \right\Vert _{2}^{2}& =\left\Vert \mathcal{P}_{\mathcal{I}_{G}}\left[
\begin{array}{c}
u \\ 
y%
\end{array}
\right] \right\Vert _{2}^{2}+\left\Vert \mathcal{P}_{\mathcal{I}_{G}^{\bot
}} \left[ 
\begin{array}{c}
u \\ 
y%
\end{array}
\right] \right\Vert _{2}^{2}  \notag \\
& =\left\Vert v\right\Vert _{2}^{2}+\left\Vert \mathcal{P}_{\mathcal{I}
_{G}^{\bot }}\left[ 
\begin{array}{c}
u \\ 
y%
\end{array}
\right] \right\Vert _{2}^{2}.  \label{eq3-7b}
\end{align}
\end{itemize}

\noindent As defined in (\ref{eq3-6}), $v$ is a latent variable for
constructing process variables $\left( u,y\right) $ in the nominal
operation. The fact (\ref{eq3-7a}) reveals that $v$ preserves the exact
amount of information needed for constructing the process variables $\left(
u,y\right) $ in the nominal operation. As the latent variable, $v$ is
achieved by maximally compressed mapping of $\left( u,y\right) $ and
preserves as much as possible the information in $\left( u,y\right) $ by its
reconstruction. Moreover, (\ref{eq3-7c}) and (\ref{eq3-7b}) imply the
maximal sensitivity of the evaluation function to faulty operations. This
motivates us to introduce the following definition.

\begin{definition}
	\label{Def1}The property (\ref{eq3-7a}) is called lossless information
	compression.
\end{definition}

We would like to mention that the term lossless is a well-established
concept in control theory and describes the property of a dynamic system in
the regard of energy balance and transport \citep{Schaft-2000-book}. For our
study on projection-based fault detection, lossless is adopted in the
context of information compression and expressed as preservation of the $%
\mathcal{L}_{2}$ norms of $v$ and $\left( \hat{u},\hat{y}\right) .$ As will
be showcased in the next subsection, under certain conditions, the concept
of lossless information compressing is equivalent to minimal sufficient
statistic for dynamic (control) systems in sense of mutual information.

Guided by the above results, an AE consisting of%
\begin{equation*}
\left\{ 
\begin{array}{l}
\text{encoder: }v=I_{G,0}^{\sim }\left[ 
\begin{array}{c}
u \\ 
y%
\end{array}%
\right] , \\ 
\text{decoder: }\left[ 
\begin{array}{c}
\hat{u} \\ 
\hat{y}%
\end{array}%
\right] =I_{G,0}v,%
\end{array}%
\right. 
\end{equation*}%
with $v$ as the latent variable, together with the residual evaluation
function $J,$%
\begin{equation*}
J=\left\Vert \left[ 
\begin{array}{c}
u \\ 
y%
\end{array}%
\right] -\left[ 
\begin{array}{c}
\hat{u} \\ 
\hat{y}%
\end{array}%
\right] \right\Vert _{2}^{2},
\end{equation*}%
as well as the threshold setting law (\ref{eq3-8}) would form an optimal
fault detection system that solves the formulated optimal fault detection
problem. Unfortunately, the aforementioned results are limited to LTI
systems. In the subsequent subsections, we will extend and realize the idea
for nonlinear systems and finally propose a learning scheme to train the
autoencoder designed based on the idea introduced in this subsection.


\subsection{Information theoretic view of latent variable learning}

In their celebrated review paper on representation learning \citep{BCV2013},
Bengio et al. have pointed out that by learning the latent variable as
feature vector \textquotedblleft good generalization means low
reconstruction error at test examples, while having high reconstruction
error for most other configurations". In the context of fault detection as a
one-class classification problem, this claim implies the maximal fault
detectability subject to the required FAR condition. To this aim, the
concept of minimal sufficient statistic is helpful \citep%
{TZ-ITW-2015,Geiger2021}. In the information-theoretic framework, learning
the latent variable as a minimal sufficient statistic is equivalent to
finding a latent representation that minimizes the mutual information of the
input and latent variables and simultaneously maximizes the mutual
information of the latent and the output variables \citep{Geiger2021}.
Concerning our task, this optimization issue is schematically formulated as
finding $v$ so that%
\begin{equation*}
I\left( \left[ 
\begin{array}{c}
u \\ 
y%
\end{array}%
\right] ;v\right) \rightarrow \min \text{ and }I\left( \left[ 
\begin{array}{c}
\hat{u} \\ 
\hat{y}%
\end{array}%
\right] ;v\right) \rightarrow \max .
\end{equation*}%
Here, $I\left( X;Y\right) $ represents the mutual information of two random
variables $X$ and $Y,$ which can be expressed in terms of the entropies of $X
$ and $Y,$ and joint entropy of $X$ and $Y,H(X),H(Y),H\left( X,Y\right) ,$ as%
\begin{equation*}
I\left( X;Y\right) =H(X)+H(Y)-H\left( X,Y\right) ,
\end{equation*}%
and is bounded by $H(X)$ or $H(Y)$ \citep{Cover2006}. In the sequel, we
highlight that the orthogonal projection-based reconstruction with the
latent variable $v$ solves this optimisation problem in a defined
probabilistic setting.

To begin with, we briefly review some existing definitions and results in
information theory concerning dynamic processes \citep%
{Papoulis_book,Cover2006}. Let $\left\{ x(k)\right\} _{k=-\infty }^{\infty
},\left\{ y(k)\right\} _{k=-\infty }^{\infty },x\in \mathbb{R}^{n},y\in 
\mathbb{R}^{m},$ denote discrete-time stationary stochastic processes. The
entropy rate of $\left\{ x(k)\right\} _{k=-\infty }^{\infty },$ and the
mutual information rate of $\left\{ x(k)\right\} _{k=-\infty }^{\infty
},\left\{ y(k)\right\} _{k=-\infty }^{\infty }$ represent the limits of the
average entropy and average mutual information, respectively, and are
defined by 
\begin{eqnarray*}
H_{R}\left( x\right) &:&=\lim_{n\rightarrow \infty }\frac{1}{n}H\left(
x(1),\cdots ,x(n)\right) , \\
I_{R}\left( x;y\right) &:&=\lim_{n\rightarrow \infty }\frac{1}{n}I\left(
x(1),\cdots ,x(n);y(1),\cdots ,y(n)\right) .
\end{eqnarray*}%
Denote by $\Phi _{x}(\theta )$, $\Phi _{y}(\theta )$ and $\Phi _{xy}(\theta
),\Phi _{yx}(\theta ),$ $\theta \in \left[ -\pi ,\pi \right] ,$ the power
spectral densities and cross-power spectral densities of $x$ and $y,$
respectively. Note that 
\begin{equation*}
\Phi _{xy}(\theta )=\Phi _{yx}^{T}(-\theta )=\Phi _{xy}^{\sim }(\theta ).
\end{equation*}

\begin{Rem}
Power and cross-power spectral densities $\Phi _{x}(\theta )/\Phi
_{y}(\theta )$ and $\Phi _{xy}(\theta )$ are the discrete-time Fourier
transform of the corresponding autocorrelation and cross-correlation
functions. Variable $\theta $ denotes the frequency \citep{Boashashbook2015}.
\end{Rem}

We are in the position to introduce the major result in this subsection. The
following lemmas are essential for our work.

\begin{lemma}
\label{Le1}\citep{Cover2006,ISHII2011} Given a zero-mean asymptotically
stationary Gaussian process $\left\{ x(k)\right\} _{k=-\infty }^{\infty
},x\in \mathbb{R}^{n},$ with power spectral density $\Phi _{x}(\theta ),$
then 
\begin{equation}
H_{R}\left( x\right) =\frac{1}{2}\log \left( 2\pi e\right) ^{n}+\frac{1}{%
4\pi }\int\limits_{-\pi }^{\pi }\log \det \Phi _{x}(\theta )d\theta .
\end{equation}
\end{lemma}

\begin{lemma}
\label{Le2}\citep{Pinsker1964,Stoorvogel1996,ZHANG2005} Given two zero-mean
asymptotically stationary Gaussian processes $\left\{ x(k)\right\}
_{k=-\infty }^{\infty },\left\{ y(k)\right\} _{k=-\infty }^{\infty },x\in 
\mathbb{R}^{n},y\in \mathbb{R}^{m},$ and the joint Gaussian process $%
\varsigma (k)=\left[ 
\begin{array}{c}
x(k) \\ 
y(k)%
\end{array}%
\right] \in \mathbb{R}^{n+m}$ with power spectral densities $\Phi
_{x}(\theta ),\Phi _{y}(\theta )$ as well as $\Phi _{\varsigma }(\theta ),$
then mutual information rate of $x$ and $y$ is given by%
\begin{equation}
I_{R}\left( x;y\right) =\frac{1}{4\pi }\int\limits_{-\pi }^{\pi }\log \frac{%
\det \Phi _{x}(\theta )\det \Phi _{y}(\theta )}{\det \Phi _{\varsigma
}(\theta )}d\theta .  \label{eq3-9}
\end{equation}
\end{lemma}
In order to fit the above information theoretic setting, we now consider,
analog to continuous-time systems, discrete-time system models and the
associated RCF and LCF as well as the normalized SIR and SKR,%
\begin{equation*}
I_{G,0}(z)=\left[ 
\begin{array}{c}
M_{0}(z) \\ 
N_{0}(z)%
\end{array}%
\right] ,K_{G,0}(z)=\left[ 
\begin{array}{cc}
-\hat{N}_{0}(z) & \text{ }\hat{M}_{0}(z)%
\end{array}%
\right] ,
\end{equation*}%
the orthogonal projection $\mathcal{P}_{\mathcal{I}_{G}}$, as described in
the previous subsection. The problem under consideration is formulated as
follows: given system%
\begin{align}
\left[ 
\begin{array}{c}
u \\ 
y%
\end{array}%
\right] & =I_{G,0}\bar{v}+K_{G,0}^{\sim }q,v=I_{G,0}^{\sim }\left[ 
\begin{array}{c}
u \\ 
y%
\end{array}%
\right] +\xi ,  \label{eq3-30a} \\
\text{ }\left[ 
\begin{array}{c}
\hat{u} \\ 
\hat{y}%
\end{array}%
\right] & =I_{G,0}v,K_{G,0}^{\sim }=\left[ 
\begin{array}{c}
-\hat{N}_{0}^{\sim } \\ 
\hat{M}_{0}^{\sim }%
\end{array}%
\right] ,  \label{eq3-30b}
\end{align}%
where $\bar{v}\in \mathcal{L}_{2}$ is a deterministic signal and represents
nominal operations$,$ 
$q,\;\xi $ are zero-mean stationary Gaussian processes. 
and $K_{G,0}^{\sim }q$ and $\xi $ are independent, find $I_{R}\left( \left[ 
\begin{array}{c}
u \\ 
y%
\end{array}%
\right] ;v\right) $ and $I_{R}\left( \left[ 
\begin{array}{c}
\hat{u} \\ 
\hat{y}%
\end{array}%
\right] ;v\right) $. Recall that 
\begin{equation*}
K_{G,0}^{\sim }q\in \mathcal{I}_{G}^{\bot }
\end{equation*}%
represents uncertainties caused by disturbances or/and faulty operations.
Thus, the mutual information rate $I_{R}\left( \left[ 
\begin{array}{c}
u \\ 
y%
\end{array}%
\right] ;v\right) $ indicates amount of uncertainties in the latent variable 
$v.$

\begin{theorem}\label{The3-1}
Given system (\ref{eq3-30a})-(\ref{eq3-30b}) with the asymptotically
stationary Gaussian processes $\left[ 
\begin{array}{c}
u \\ 
y%
\end{array}%
\right] ,\left[ 
\begin{array}{c}
\hat{u} \\ 
\hat{y}%
\end{array}%
\right] $ and $v,$ it holds%
\begin{equation}
I_{R}\left( \left[ 
\begin{array}{c}
u \\ 
y%
\end{array}%
\right] ;v\right) =0,H_{R}\left( v\right) =H_{R}\left( \left[ 
\begin{array}{c}
\hat{u} \\ 
\hat{y}%
\end{array}%
\right] \right) .  \label{eq3-35}
\end{equation}
\end{theorem}

\begin{proof}
Let 
\begin{align*}
\beta & =\left[ 
\begin{array}{c}
u \\ 
y%
\end{array}%
\right] -\mathbb{E}\left( \left[ 
\begin{array}{c}
u \\ 
y%
\end{array}%
\right] \right) =\left[ 
\begin{array}{c}
u \\ 
y%
\end{array}%
\right] -I_{G,0}\bar{v}=K_{G,0}^{\sim }q, \\
\alpha & =v-\mathbb{E}\left( v\right) =I_{G,0}^{\sim }K_{G,0}^{\sim }q+\xi ,
\end{align*}%
where $\mathbb{E}\left( \cdot \right) $ denotes expectation. Hence, both $%
\alpha $ and $\beta $ are zero-mean asymptotically stationary Gaussian
processes. It follows from Lemma \ref{Le2} that%
\begin{gather}
I_{R}\left( \left[ 
\begin{array}{c}
u \\ 
y%
\end{array}%
\right] ;v\right) =\frac{1}{4\pi }\int\limits_{-\pi }^{\pi }\log \frac{\det
\Phi _{\beta }(\theta )\det \Phi _{\alpha }(\theta )}{\det \Phi _{\gamma
}(\theta )}d\theta , \\
\gamma =\left[ 
\begin{array}{c}
\beta \\ 
\alpha%
\end{array}%
\right] ,\Phi _{\gamma }(\theta )=\left[ 
\begin{array}{cc}
\Phi _{\beta }(\theta ) & \Phi _{\beta \alpha }(\theta ) \\ 
\Phi _{\alpha \beta }(\theta ) & \Phi _{\alpha }(\theta )%
\end{array}%
\right] , \\
\Phi _{\beta \alpha }(\theta )=\Phi _{\beta }(\theta )I_{G,0}(\theta ).
\end{gather}%
Since%
\begin{equation*}
\Phi _{\beta }(\theta )=K_{G,0}^{\sim }(\theta )\Phi _{q}(\theta
)K_{G,0}(\theta ),K_{G,0}(\theta )I_{G,0}(\theta )=0,
\end{equation*}%
and $K_{G,0}^{\sim }q$ and $\xi $ are independent, it turns out 
\begin{gather*}
\Phi _{\beta \alpha }(\theta )=\Phi _{\beta }(\theta )I_{G,0}(\theta )=\Phi
_{\alpha \beta }^{\sim }(\theta )=0\Longrightarrow \\
\log \frac{\det \Phi _{\beta }(\theta )\det \Phi _{\alpha }(\theta )}{\det
\Phi _{\gamma }(\theta )}=0\Longrightarrow I_{R}\left( \left[ 
\begin{array}{c}
u \\ 
y%
\end{array}%
\right] ;v\right) =0.
\end{gather*}%
Recall (\ref{eq3-33}), namely 
\begin{equation*}
\left[ 
\begin{array}{c}
I_{G,0}^{\sim } \\ 
K_{G,0}%
\end{array}%
\right] \left[ 
\begin{array}{c}
I_{G,0}^{\sim } \\ 
K_{G,0}%
\end{array}%
\right] ^{\sim }=\left[ 
\begin{array}{cc}
I & \text{ }0 \\ 
0 & \text{ }I%
\end{array}%
\right] ,
\end{equation*}%
which means that $\left[ 
\begin{array}{c}
I_{G,0}^{\sim } \\ 
K_{G,0}%
\end{array}%
\right] $ is a unitary mapping. Multiplying $\left[ 
\begin{array}{c}
\hat{u} \\ 
\hat{y}%
\end{array}%
\right] $ by $\left[ 
\begin{array}{c}
I_{G,0}^{\sim } \\ 
K_{G,0}%
\end{array}%
\right] $ yields%
\begin{equation*}
\left[ 
\begin{array}{c}
I_{G,0}^{\sim } \\ 
K_{G,0}%
\end{array}%
\right] \left[ 
\begin{array}{c}
\hat{u} \\ 
\hat{y}%
\end{array}%
\right] =\left[ 
\begin{array}{c}
I_{G,0}^{\sim } \\ 
K_{G,0}%
\end{array}%
\right] I_{G,0}v=\left[ 
\begin{array}{c}
v \\ 
0%
\end{array}%
\right] .
\end{equation*}%
As a result, by Lemma \ref{Le1}%
\begin{gather}
H_{R}\left( \left[ 
\begin{array}{c}
\hat{u} \\ 
\hat{y}%
\end{array}%
\right] \right) =H_{R}\left( v\right)  \label{eq3-36} \\
=\frac{1}{2}\log \left( 2\pi e\right) ^{p}+\frac{1}{4\pi }\int\limits_{-\pi
}^{\pi }\log \det \Phi _{\alpha }(\theta )d\theta .  \notag
\end{gather}%
The theorem is thus proven.
\end{proof}

Note that the relation (\ref{eq3-36}) implies that $I_{R}\left( \left[ 
\begin{array}{c}
\hat{u} \\ 
\hat{y}%
\end{array}%
\right] ;v\right) $ is equal to its upper-bound and thus reaches the
maximum. In summary, Theorem \ref{The3-1} gives a solution to the optimal
selection of the latent variable and proves that the latent variable $%
v=I_{G,0}^{\sim }\left[ 
\begin{array}{c}
u \\ 
y%
\end{array}%
\right] $ is a minimal sufficient statistic on the assumption of the system
model (\ref{eq3-30a})-(\ref{eq3-30b}). In particular, the relation (\ref%
{eq3-35}) provides us with an information theoretic interpretation for the
concept of lossless information compression given in Definition \ref{Def1}.
We would like to emphasize that the property (\ref{eq3-35}) is the result of
the adopted orthogonal projection and leads to the maximal fault
detectability thanks to the maximal sensitivity of the evaluation function
to the faulty operations.

\subsection{Extension to nonlinear dynamic systems}

We now consider nonlinear systems given in (\ref{eq2-1a}). For the sake of
simplicity, we restrict our study to a class of nonlinear systems, the
so-called affine systems modelled by the following state space representation%
\begin{equation}
\Sigma :%
\begin{cases}
\dot{x}=a(x)+B(x)u \\ 
y=c(x)+D(x)u,%
\end{cases}
\label{eq3-10}
\end{equation}%
where $a(x),B(x),c(x),D(x)$ are smooth functions of appropriate dimensions.
Analog to LTI systems, the definition of stable image representation of $%
\Sigma $ is essential for our extension effort.

\begin{definition}
	\citep{Schaft-2000-book} Given system $\Sigma $ defined by (\ref{eq3-10}),
	system $\Sigma _{I}:\mathcal{V}\rightarrow \mathcal{L}_{2}\times \mathcal{L}%
	_{2}$ is an SIR of the system $\Sigma $ if for all $u\in \mathcal{U\subset L}%
	_{2}$ and $y=\Sigma (u)\in \mathcal{Y\subset L}_{2}$, there exists $v\in 
	\mathcal{V\subset L}_{2}$ such that 
	\begin{equation*}
	\left[ 
	\begin{array}{c}
	u \\ 
	y%
	\end{array}%
	\right] =\Sigma _{I}\left( v\right).
	\end{equation*}%
\end{definition}
A state space representation of the SIR is given by 
\begin{gather}
\Sigma _{I}:\hspace{-2pt}%
\begin{cases}
\dot{x}=a(x)+B(x)g(x)+B(x)V(x)v\hspace{-2pt} \\ 
\;\;\;=\bar{a}(x)+B(x)V(x)v \\ 
\left[ 
\begin{array}{c}
u \\ 
y%
\end{array}%
\right] \hspace{-2pt}=\hspace{-2pt}\bar{c}(x)\hspace{-2pt}+\hspace{-2pt}\bar{%
D}(x)V(x)v,%
\end{cases}
\label{eq3-11} \\
\bar{c}(x)=\left[ 
\begin{array}{c}
g(x) \\ 
c(x)+D(x)g(x)%
\end{array}%
\right] ,\bar{D}(x)=\left[ 
\begin{array}{c}
V(x) \\ 
D(x)V(x)%
\end{array}%
\right] ,  \notag
\end{gather}%
where $g(x)$ is designed such that $\bar{a}(x)=a(x)+B(x)g(x)$ is
asymptotically stable and $V(x)\in \mathbb{R}^{p\times p}$ is invertible 
\citep{Scherpen1994,Ball1996}. In the context of feedback control systems, 
\begin{equation*}
u=g(x)+V(x)v
\end{equation*}%
can be understood as a controller with state feedback $g(x),$ feed-forward
controller $V(x)$ and ${v}$ as the reference signal. Note that 
\begin{equation*}
\bar{D}^{T}(x)\bar{D}(x)=V^{T}(x)\left( I+D^{T}(x)D(x)\right) V(x)
\end{equation*}%
is invertible.

Based on the SIR of $\Sigma ,$ the image subspace of $\Sigma ,$ which
defines the set of the process data generated under the nominal operation
conditions, is defined as follows.

\begin{definition}
	Given system $\Sigma $, its SIR $\Sigma _{I},$ and $u\in \mathcal{U},y\in 
	\mathcal{Y},$%
	\begin{equation}
	\mathcal{I}_{\Sigma }=\left\{ \left[ 
	\begin{array}{c}
	u \\ 
	y%
	\end{array}%
	\right] :\left[ 
	\begin{array}{c}
	u \\ 
	y%
	\end{array}%
	\right] =\Sigma _{I}\left( v\right) ,v\in \mathcal{V}\right\} 
	\end{equation}%
	is called image subspace of $\Sigma .$
\end{definition}

Recall that for LTI systems, the basic idea of constructing an optimal fault
detection system is to find an operator $\Pi ,$%
\begin{equation*}
\left[ 
\begin{array}{c}
\hat{u} \\ 
\hat{y}%
\end{array}%
\right] =\Pi \left[ 
\begin{array}{c}
u \\ 
y%
\end{array}%
\right] ,
\end{equation*}%
which is idempotent, i.e.%
\begin{equation}
\Pi \left[ 
\begin{array}{c}
\hat{u} \\ 
\hat{y}%
\end{array}%
\right] =\Pi \left( \Pi \left[ 
\begin{array}{c}
u \\ 
y%
\end{array}%
\right] \right) =\Pi \left[ 
\begin{array}{c}
u \\ 
y%
\end{array}%
\right] ,  \label{eq3-12a}
\end{equation}%
and guarantees lossless information compression, as defined in Definition %
\ref{Def1},%
\begin{equation}
\left\Vert v\right\Vert _{2}=\left\Vert \left[ 
\begin{array}{c}
\hat{u} \\ 
\hat{y}%
\end{array}%
\right] \right\Vert _{2}=\left\Vert \Pi \left[ 
\begin{array}{c}
u \\ 
y%
\end{array}%
\right] \right\Vert _{2}.  \label{eq3-12}
\end{equation}%
For LTI systems, the normalized SIR serves for this end with $\Pi
=I_{G,0}I_{G,0}^{\sim },$ which cannot be applied to nonlinear systems.
Notice that the conditions (\ref{eq3-12a})-(\ref{eq3-12}) are essential for
the orthogonal projection. It inspires us to address the issues for
nonlinear systems under these two aspects. To this end, the concepts of the
Hamiltonian extension and inner systems are firstly introduced, which are
well established in nonlinear control theory \citep%
{Schaft-book1987,Scherpen1994}.

Given system (\ref{eq3-11}), the Hamiltonian extension of $\Sigma _{I}$ is a
dynamic system described by%
\begin{equation}
\begin{cases}
\dot{x}=\bar{a}(x)+B(x)v \\ 
\dot{p}=-\left( \frac{\partial \bar{a}(x)}{\partial x}+\frac{\partial B(x)}{%
\partial x}v\right) ^{T}p-\left( \frac{\partial \bar{c}\left( x\right) }{%
\partial x}+\frac{\partial \bar{D}(x)}{\partial x}\right) ^{T}v_{a}, \\ 
\bar{z}:=\left[ 
\begin{array}{c}
u \\ 
y%
\end{array}%
\right] =\bar{c}(x)+\bar{D}(x)v \\ 
z_{a}=B^{T}(x)p+\bar{D}^{T}(x)v_{a},z_{a}\in \mathbb{R}^{p},v_{a}\in \mathbb{%
\ R}^{m+p}%
\end{cases}
\label{eq3-15}
\end{equation}%
with state variables $\left( x,p\right) ,p\in \mathbb{\ R}^{n},$ input
variables $\left( v,v_{a}\right) $ and output variable $\left( \bar{z}%
,z_{a}\right) $. Define the Hamiltonian function 
\begin{align*}
H\left( x,p,v\right) =& \frac{1}{2}\left( \bar{c}(x)+\bar{D}(x)v\right)
^{T}\left( \bar{c}(x)+\bar{D}(x)v\right)  \\
& +p^{T}\left( \bar{a}(x)+B(x)v\right) 
\end{align*}%
and connect $\bar{z}$ and $v_{a},v_{a}=\bar{z}=\bar{c}(x)+\bar{D}(x)v.$ We
have the following Hamiltonian system 
\begin{equation*}
z_{a}=\left( D\Sigma _{I}\right) ^{T}\circ \Sigma _{I}\left( v\right) ,
\end{equation*}%
whose state space representation can be written in the compact form%
\begin{equation}
\left( D\Sigma _{I}\right) ^{T}\circ \Sigma _{I}:%
\begin{cases}
\dot{x}=\frac{\partial H}{\partial p}\left( x,p,v\right)  \\ 
\dot{p}=-\frac{\partial H}{\partial x}\left( x,p,v\right)  \\ 
z_{a}=\frac{\partial H}{\partial v}\left( x,p,v\right) .%
\end{cases}
\label{eq3-13}
\end{equation}%
We now introduce the definition of inner systems \citep{Scherpen1994}.

\begin{definition}
Given nonlinear affine system (\ref{eq3-11}) and the corresponding
Hamiltonian system (\ref{eq3-13}), $\Sigma _{I}$ is inner if 
\begin{equation}
z_{a}=v,  \label{eq3-14a}
\end{equation}%
and there exists a storage function $P(x)\geq 0, P(0)=0$ so that 
\begin{equation}
P\left( x\left( t_{2}\right) \right) -P\left( x\left( t_{1}\right) \right)
=\int_{t_{1}}^{t_{2}}\left( \frac{1}{2}v^{T}v-\frac{1}{2}\bar{z}^{T}\bar{z}%
\right) d\tau .  \label{eq3-14b}
\end{equation}
\end{definition}

\begin{Rem}
A system with the property (\ref{eq3-14b}) is called lossless with respect
to the supply rate 
\begin{equation*}
s(v,\bar{z})=\frac{1}{2}v^{T}v-\frac{1}{2}\bar{z}^{T}\bar{z}.
\end{equation*}%
For $t_{1}=0,t_{2}=\infty ,x\left( 0\right) =0,$ it holds 
\begin{gather*}
\int_{0}^{\infty }\left( \frac{1}{2}v^{T}v-\frac{1}{2}\bar{z}^{T}\bar{z}%
\right) d\tau =0 \\
\Longleftrightarrow\left\Vert v\right\Vert _{2}=\left\Vert \bar{z}%
\right\Vert _{2}=\left\Vert \left[ 
\begin{array}{c}
u \\ 
y%
\end{array}%
\right] \right\Vert _{2}.
\end{gather*}
\end{Rem}

The existence conditions for the system (\ref{eq3-11}) to be inner are
summarized in the following lemma.

\begin{lemma}
	\label{lemma1}The nonlinear affine system (\ref{eq3-11}) is inner, if there
	exists $P(x)\geq 0$ such that the following equations are feasible 
	\begin{align}
	& P_{x}\left( x\right) \bar{a}(x)+\frac{1}{2}\bar{c}^{T}(x)\bar{c}(x)=0,
	\label{eq3-14a} \\
	& P_{x}\left( x\right) B(x)+\bar{c}^{T}(x)\bar{D}(x)=0, \\
	& \bar{D}^{T}(x)\bar{D}(x)=I,  \label{eq3-14c}
	\end{align}%
	where $P_{x}(x)=\frac{\partial P}{\partial x}(x)$. 
\end{lemma}

The results given in the above lemma are well-known, see, for instance, \cite%
{Scherpen1994,Ball1996}. Hence, the proof of Lemma \ref{lemma1} is omitted.

According to Lemma \ref{lemma1}, the SIR $\Sigma _{I}$ becomes inner, when
feedback and feed-forward controllers $g(x)$ and $V(x)$ are set in such a
way such that equations (\ref{eq3-14a})-(\ref{eq3-14c}) are solved.

Let $\Sigma _{I}$ be inner, and denoted by $\Sigma _{I,0}.$ We propose to
build estimator $\Pi $ as follows%
\begin{equation}
\left[ 
\begin{array}{c}
\hat{u} \\ 
\hat{y}%
\end{array}%
\right] =\Pi \left[ 
\begin{array}{c}
u \\ 
y%
\end{array}%
\right] =\Sigma _{I,0}\circ \left( D\Sigma _{I,0}\right) ^{T}\left( \left[ 
\begin{array}{c}
u \\ 
y%
\end{array}%
\right] \right) ,  \label{eq3-16a}
\end{equation}%
whose state space representation is, by connecting $z_{a}$ and $v$ in the
Hamiltonian extension (\ref{eq3-15}) and defining%
\begin{equation*}
v_{a}=\left[ 
\begin{array}{c}
u \\ 
y%
\end{array}%
\right] ,
\end{equation*}%
given by 
\begin{align}
& \;\;\;\;\;\;\;\;\;\;\Sigma _{I,0}\circ \left( D\Sigma _{I,0}\right) ^{T}: 
\notag \\
& 
\begin{cases}
\dot{x}=\bar{a}(x)+B(x)B^{T}(x)p+B(x)\bar{D}^{T}(x)\left[ 
\begin{array}{c}
u \\ 
y%
\end{array}%
\right]  \\ 
\dot{p}=-\left( \frac{\partial \bar{a}(x)}{\partial x}+\frac{\partial }{%
\partial x}B(x)\left( B^{T}(x)p+\bar{D}^{T}(x)\left[ 
\begin{array}{c}
u \\ 
y%
\end{array}%
\right] \right) \right) ^{T}p \\ 
\;\;\;-\hspace{-2pt}\left( \hspace{-2pt}\frac{\partial \bar{c}\left(
x\right) }{\partial x}\right) ^{T}\hspace{-3pt}\left[ 
\begin{array}{c}
u \\ 
y%
\end{array}%
\right] \hspace{-3pt}-\hspace{-3pt}\left( B^{T}(x)p\hspace{-2pt}+\hspace{-2pt%
}\bar{D}^{T}(x)\hspace{-2pt}\left[ 
\begin{array}{c}
u \\ 
y%
\end{array}%
\right] \hspace{-2pt}\right) ^{T}\hspace{-4pt}\frac{\partial ^{T}\bar{D}(x)}{%
\partial x}\hspace{-2pt}\left[ 
\begin{array}{c}
u \\ 
y%
\end{array}%
\right]  \\ 
\left[ 
\begin{array}{c}
\hat{u} \\ 
\hat{y}%
\end{array}%
\right] =\bar{c}(x)+\bar{D}(x)B^{T}(x)p+\bar{D}(x)\bar{D}^{T}(x)\left[ 
\begin{array}{c}
u \\ 
y%
\end{array}%
\right] .%
\end{cases}
\label{eq3-16}
\end{align}%
Below, it is demonstrated that system $\Pi =\Sigma _{I,0}\circ \left(
D\Sigma _{I,0}\right) ^{T}$ is idempotent and lossless by information
compression, i.e. satisfying (\ref{eq3-12}).

\begin{theorem}
	\label{Theo1}Given system (\ref{eq3-16a}), it holds%
	\begin{align*}
	&\Sigma _{I,0}\circ \left( D\Sigma _{I,0}\right) ^{T}\circ \Sigma _{I,0}\circ
	\left( D\Sigma _{I,0}\right) ^{T}\left( \left[ 
	\begin{array}{c}
	u \\ 
	y%
	\end{array}%
	\right] \right) \\
	&=\Sigma _{I,0}\circ \left( D\Sigma _{I,0}\right) ^{T}\left( %
	\left[ 
	\begin{array}{c}
	u \\ 
	y%
	\end{array}%
	\right] \right) , \\
	&\left\Vert \left( D\Sigma _{I,0}\right) ^{T}\left( \left[ 
	\begin{array}{c}
	u \\ 
	y%
	\end{array}%
	\right] \right) \right\Vert _{2}=\left\Vert v\right\Vert _{2}=\left\Vert %
	\left[ 
	\begin{array}{c}
	\hat{u} \\ 
	\hat{y}%
	\end{array}%
	\right] \right\Vert _{2}\\
	&=\left\Vert \Sigma _{I,0}\circ \left( D\Sigma
	_{I,0}\right) ^{T}\left( \left[ 
	\begin{array}{c}
	u \\ 
	y%
	\end{array}%
	\right] \right) \right\Vert _{2}.
	\end{align*}
\end{theorem}

\begin{proof}
Let 
\begin{equation*}
v=\left( D\Sigma _{I,0}\right) ^{T}\left( \left[ 
\begin{array}{c}
u \\ 
y%
\end{array}%
\right] \right) .
\end{equation*}%
Since $\Sigma _{I,0}$ is inner, it holds 
\begin{gather*}
\Sigma _{I,0}\circ \left( D\Sigma _{I,0}\right) ^{T}\circ \Sigma _{I,0}\circ
\left( D\Sigma _{I,0}\right) ^{T}\left( \left[ 
\begin{array}{c}
u \\ 
y%
\end{array}%
\right] \right) \\
=\Sigma _{I,0}\circ \left( D\Sigma _{I,0}\right) ^{T}\left( \left[ 
\begin{array}{c}
u \\ 
y%
\end{array}%
\right] \right) .
\end{gather*}%
Moreover, due to the lossless property of $\Sigma _{I,0}$, 
\begin{equation*}
\left\Vert v\right\Vert _{2}=\left\Vert \Sigma _{I,0}(v)\right\Vert _{2},
\end{equation*}%
which results in 
\begin{gather*}
\left\Vert \left( D\Sigma _{I,0}\right) ^{T}\left( \left[ 
\begin{array}{c}
u \\ 
y%
\end{array}%
\right] \right) \right\Vert _{2}=\left\Vert v\right\Vert _{2}=\left\Vert
\Sigma _{I,0}(v)\right\Vert _{2} \\
=\left\Vert \Sigma _{I,0}\circ \left( D\Sigma _{I,0}\right) ^{T}\left( \left[
\begin{array}{c}
u \\ 
y%
\end{array}%
\right] \right) \right\Vert _{2}=\left\Vert \left[ 
\begin{array}{c}
\hat{u} \\ 
\hat{y}%
\end{array}%
\right] \right\Vert _{2}.
\end{gather*}%
The theorem is thus proven.
\end{proof}

In the context of fault detection, Theorem \ref{Theo1} implies that

\begin{itemize}
\item $\forall \left[ 
\begin{array}{c}
u \\ 
y%
\end{array}
\right] \in \mathcal{I}_{\Sigma },$ i.e. process data generated under the
nominal operation condition, 
\begin{align*}
\left[ 
\begin{array}{c}
\hat{u} \\ 
\hat{y}%
\end{array}
\right] & =\Pi \left[ 
\begin{array}{c}
u \\ 
y%
\end{array}
\right] =\Sigma _{I,0}(v)=\left[ 
\begin{array}{c}
u \\ 
y%
\end{array}
\right] , \\
v& =\left( D\Sigma _{I,0}\right) ^{T}\left( \left[ 
\begin{array}{c}
u \\ 
y%
\end{array}
\right] \right) ,
\end{align*}

\item the generation of the latent variable by 
\begin{equation*}
v=\left( D\Sigma _{I,0}\right) ^{T}\left( \left[ 
\begin{array}{c}
u \\ 
y%
\end{array}
\right] \right)
\end{equation*}
is a lossless information compression, and

\item $\forall \left[ 
\begin{array}{c}
u \\ 
y%
\end{array}
\right] \notin \mathcal{I}_{\Sigma },$ 
\begin{gather*}
\left\Vert \left[ 
\begin{array}{c}
u \\ 
y%
\end{array}
\right] \right\Vert _{2}\geq \left\Vert \left[ 
\begin{array}{c}
\hat{u} \\ 
\hat{y}%
\end{array}
\right] \right\Vert _{2}=\left\Vert v\right\Vert _{2}, \\
\Longrightarrow \left\Vert \left[ 
\begin{array}{c}
u \\ 
y%
\end{array}
\right] -\left[ 
\begin{array}{c}
\hat{u} \\ 
\hat{y}%
\end{array}
\right] \right\Vert _{2}\geq 0.
\end{gather*}
\end{itemize}

\noindent As a result, the formulated optimal fault detection problem can be
solved by (i) constructing an AE realizing%
\begin{equation}
\left\{ 
\begin{array}{l}
\text{encoder: }v=\left( D\Sigma _{I,0}\right) ^{T}\left( \left[ 
\begin{array}{c}
u \\ 
y%
\end{array}%
\right] \right) , \\ 
\text{decoder: }\left[ 
\begin{array}{c}
\hat{u} \\ 
\hat{y}%
\end{array}%
\right] =\Sigma _{I,0}(v),%
\end{array}%
\right.   \label{eq3-17}
\end{equation}%
with $v$ as the latent variable, (ii) defining the evaluation function%
\begin{equation}
J=\left\Vert \left[ 
\begin{array}{c}
u \\ 
y%
\end{array}%
\right] -\left[ 
\begin{array}{c}
\hat{u} \\ 
\hat{y}%
\end{array}%
\right] \right\Vert _{2}^{2}=\left\Vert \left[ 
\begin{array}{c}
u \\ 
y%
\end{array}%
\right] -\Sigma _{I,0}(v)\right\Vert _{2}^{2},  \label{eq3-18}
\end{equation}%
and setting the threshold as 
\begin{align}
J_{th}& =\sup_{\substack{ \left[ 
\begin{array}{c}
u \\ 
y%
\end{array}%
\right] \in \mathcal{Z} \\ FAR\leq \gamma }}dist\left( \left[ 
\begin{array}{c}
u \\ 
y%
\end{array}%
\right] ,\mathcal{I}_{\Sigma }\right) ^{2},\mathcal{Z}=\mathcal{U\times Y} 
\notag \\
& =\sup_{\substack{ \left[ 
\begin{array}{c}
u \\ 
y%
\end{array}%
\right] \in \mathcal{Z} \\ FAR\leq \gamma }}\inf_{\left[ 
\begin{array}{c}
u_{0} \\ 
y_{0}%
\end{array}%
\right] \in \mathcal{I}_{\Sigma }}\left\Vert \left[ 
\begin{array}{c}
u \\ 
y%
\end{array}%
\right] -\left[ 
\begin{array}{c}
u_{0} \\ 
y_{0}%
\end{array}%
\right] \right\Vert _{2}^{2}.  \label{eq3-19}
\end{align}%
It is obvious that the above solution can only be analytically achieved if
(i) the process model (\ref{eq3-10}) exists, and (ii) $g(x)$ and $V(x)$ in
the SIR $\Sigma _{I,0}$ can be determined by solving equations (\ref{eq3-14a}%
)-(\ref{eq3-14c}). In practical applications, both of these requirements can
often not be satisfied or are satisfied at remarkably high engineering
costs. In the next subsection, we are going to introduce a data-driven
solution alternatively using an AE that will be learnt in such a way that it
is idempotent and results in lossless information compression via the latent
variable $v$.

\subsection{An AE-based fault detection scheme \label{sub4}}

We now propose a learning algorithm to train an AE to realize the optimal
fault detection system given in (\ref{eq3-17})-(\ref{eq3-19}) on the
assumption that

\begin{itemize}
\item sufficient process data $\left( u,y\right) $ have been collected
during fault-free operations,

\item the data have been recorded batchwise, 
\begin{equation*}
\left[ 
\begin{array}{c}
u^{(i)} \\ 
y^{(i)}%
\end{array}%
\right] =\left\{ \left[ 
\begin{array}{c}
u(k) \\ 
y(k)%
\end{array}%
\right] ,k\in \left[ k_{i},k_{i}+N\right] ,\right\} ,
\end{equation*}%
where $k,k_{i}$ are the sampling number, $N$ is the length of the data
batch, and

\item the data set is denoted by $\mathcal{Z},$ 
\begin{equation*}
\mathcal{Z}=\left\{ \left[ 
\begin{array}{c}
u^{(i)} \\ 
y^{(i)}%
\end{array}%
\right] ,i=1,\cdots M\right\} .
\end{equation*}
\end{itemize}

\noindent For our purpose,

\begin{itemize}
\item recurrent neural networks $(\mathcal{R}\mathcal{N}\mathcal{N}s)$ are
firstly constructed as 
\begin{equation}
\left\{ 
\begin{array}{l}
\hspace{-2pt}\text{encoder: }\mathcal{R}\mathcal{N}\mathcal{N}_{en}\hspace{%
-2pt}\left( \hspace{-1pt}\theta _{en},\left[ 
\begin{array}{c}
u \\ 
y%
\end{array}%
\right] \hspace{-1pt}\right) \hspace{-3pt}=\hspace{-2pt}\left( D\Sigma
_{I,0}\right) ^{T}\hspace{-2pt}\left( \left[ 
\begin{array}{c}
u \\ 
y%
\end{array}%
\right] \right) \hspace{-3pt}=\hspace{-2pt}v \\ 
\hspace{-2pt}\text{decoder: }\mathcal{R}\mathcal{N}\mathcal{N}_{de}\left(
\theta _{de},v\right) =\Sigma _{I,0}(v)=\left[ 
\begin{array}{c}
\hat{u} \\ 
\hat{y}%
\end{array}%
\right] ,%
\end{array}%
\right. 
\end{equation}%
and they build an AE, based on it,

\item the evaluation function is computed (during training and later for the
online implementation) 
\begin{align*}
J& =\left\Vert \left[ 
\begin{array}{c}
u \\ 
y%
\end{array}%
\right] -\left[ 
\begin{array}{c}
\hat{u} \\ 
\hat{y}%
\end{array}%
\right] \right\Vert _{2}^{2}=\left\Vert \left[ 
\begin{array}{c}
u \\ 
y%
\end{array}%
\right] -\mathcal{R}\mathcal{N}\mathcal{N}_{de}\left( \theta _{de},v\right)
\right\Vert _{2}^{2} \\
& =\left\Vert \left[ 
\begin{array}{c}
u \\ 
y%
\end{array}%
\right] -\mathcal{R}\mathcal{N}\mathcal{N}_{de}\left( \theta _{de},\mathcal{R%
}\mathcal{N}\mathcal{N}_{en}\left( \theta _{en},\left[ 
\begin{array}{c}
u \\ 
y%
\end{array}%
\right] \right) \right) \right\Vert _{2}^{2},
\end{align*}%
after the AE is learnt, and finally

\item the threshold is to be determined as 
\begin{align*}
&J_{th}=\sup_{\substack{ \left[ 
\begin{array}{c}
u^{(i)} \\ 
y^{(i)}%
\end{array}
\right] \in \mathcal{Z}  \\ FAR\leq \gamma }}\left\Vert \left[ 
\begin{array}{c}
u^{(i)} \\ 
y^{(i)}%
\end{array}
\right] -\left[ 
\begin{array}{c}
\hat{u}^{(i)} \\ 
\hat{y}^{(i)}%
\end{array}
\right]\right\Vert _{2}^{2} \\
&\left[ 
\begin{array}{c}
\hat{u}^{(i)} \\ 
\hat{y}^{(i)}%
\end{array}
\right]=\mathcal{R}\mathcal{N}\mathcal{N}_{de}\left( \theta _{de},\mathcal{R 
}\mathcal{N}\mathcal{N}_{en}\left( \theta _{en},\left[ 
\begin{array}{c}
u^{(i)} \\ 
y^{(i)}%
\end{array}
\right] \right) \right) .
\end{align*}
\end{itemize}

\noindent We would like to remark that process data $\left( u,y\right) $
exist generally as discrete-time samples. Correspondingly, the $\mathcal{L}%
_{2}$ norm adopted in the evaluation function is approximated by%
\begin{equation*}
\left\Vert \alpha \right\Vert
_{2}^{2}:=\sum\limits_{i=k_{0}}^{k_{0}+N}\alpha ^{T}(i)\alpha (i).
\end{equation*}%
Next, as a major contribution of our work, the loss function for optimizing
the parameters $\theta _{en},\theta _{de}$ of $\mathcal{R}\mathcal{N}%
\mathcal{N}s$ is defined. In order to learn an idempotent autoencoder with
lossless information compression, two regularized terms, $\mathcal{L}%
_{2}(\theta _{de},\theta _{en})$ and $\mathcal{L}_{3}(\theta _{de},\theta
_{en}),$ are introduced into the loss function, in addition to the standard
index, 
\begin{equation}
\mathcal{L}_{1}(\theta _{de},\theta _{en})=\frac{1}{M}\sum_{i=1}^{M}\left%
\Vert \left[ 
\begin{array}{c}
u^{(i)} \\ 
y^{(i)}%
\end{array}%
\right] -\left[ 
\begin{array}{c}
\hat{u}^{(i)} \\ 
\hat{y}^{(i)}%
\end{array}%
\right] \right\Vert _{2}^{2}.  \label{eq3-41}
\end{equation}%
It is obvious that $\mathcal{L}_{1}(\theta _{de},\theta _{en})$ is dedicated
to minimizing the reconstruction error and thus the evaluation function%
\begin{equation*}
J=\left\Vert \left[ 
\begin{array}{c}
u \\ 
y%
\end{array}%
\right] -\left[ 
\begin{array}{c}
\hat{u} \\ 
\hat{y}%
\end{array}%
\right] \right\Vert _{2}^{2}.
\end{equation*}%
The two regularized terms are 
\begin{gather}
\mathcal{L}_{2}(\theta _{de},\theta _{en})=  \notag \\
\frac{1}{M}\hspace{-2pt}\sum_{i=1}^{M}\hspace{-2pt}\left\Vert \left[ 
\begin{array}{c}
\hat{u}^{(i)} \\ 
\hat{y}^{(i)}%
\end{array}%
\right] \hspace{-3pt}-\hspace{-2pt}\mathcal{R}\mathcal{N}\mathcal{N}_{de}%
\hspace{-2pt}\left( \theta _{de},\mathcal{R}\mathcal{N}\mathcal{N}_{en}%
\hspace{-2pt}\left( \theta _{en},\left[ 
\begin{array}{c}
\hat{u}^{(i)} \\ 
\hat{y}^{(i)}%
\end{array}%
\right] \right) \right) \right\Vert _{2}^{2},  \label{eq3-42} \\
\mathcal{L}_{3}(\theta _{de},\theta _{en})=  \notag \\
\frac{1}{M}\hspace{-2pt}\sum_{i=1}^{M}\hspace{-2pt}\left\Vert \mathcal{R}%
\mathcal{N}\mathcal{N}_{en}\hspace{-2pt}\left( \theta _{en},\left[ 
\begin{array}{c}
u^{(i)} \\ 
y^{(i)}%
\end{array}%
\right] \right) \hspace{-3pt}-\hspace{-2pt}\mathcal{R}\mathcal{N}\mathcal{N}%
_{en}\hspace{-2pt}\left( \theta _{en},\left[ 
\begin{array}{c}
\hat{u}^{(i)} \\ 
\hat{y}^{(i)}%
\end{array}%
\right] \right) \right\Vert _{2}^{2}.  \label{eq3-43}
\end{gather}%
It is apparent that $\mathcal{L}_{2}(\theta _{de},\theta _{en})$ is
dedicated to training the AE to be idempotent, and $\mathcal{L}_{3}(\theta
_{de},\theta _{en})$ serves for the purpose 
\begin{gather*}
\mathcal{R}\mathcal{N}\mathcal{N}_{en}\hspace{-2pt}\left( \theta _{en},\left[
\begin{array}{c}
u \\ 
y%
\end{array}%
\right] \right) -\mathcal{R}\mathcal{N}\mathcal{N}_{en}\hspace{-2pt}\left(
\theta _{en},\left[ 
\begin{array}{c}
\hat{u} \\ 
\hat{y}%
\end{array}%
\right] \right)  \\
=v-\hat{v}\rightarrow 0\Longleftrightarrow \mathcal{R}\mathcal{N}\mathcal{N}%
_{en}\left( \theta _{en},\mathcal{R}\mathcal{N}\mathcal{N}_{de}\left( \theta
_{de}\right) \right) \rightarrow I,
\end{gather*}%
which implies that the AE is inner. Consequently, 
\begin{figure*}[t]
\centering\includegraphics[scale=0.5]{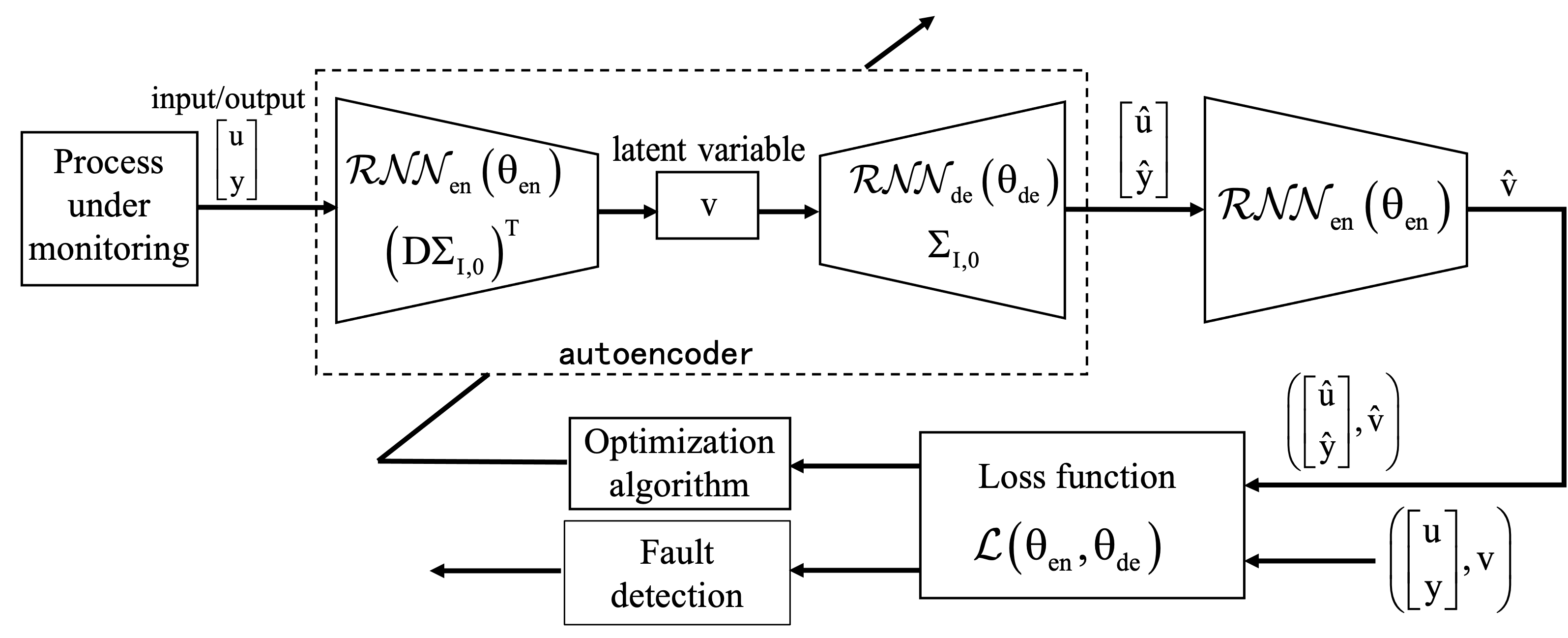}
\caption{The schematic of the inner-autoencoder-based fault detection}
\label{fig3}
\end{figure*}
\begin{align*}
& \left\Vert \left[ 
\begin{array}{c}
\hat{u} \\ 
\hat{y}%
\end{array}%
\right] \right\Vert _{2}=\left\langle v,\mathcal{R}\mathcal{N}\mathcal{N}%
_{en}\left( \theta _{en},\mathcal{R}\mathcal{N}\mathcal{N}_{de}\left( \theta
_{de},v\right) \right) \right\rangle , \\
& \left\Vert \left[ 
\begin{array}{c}
\hat{u} \\ 
\hat{y}%
\end{array}%
\right] \right\Vert _{2}-\left\Vert v\right\Vert _{2}\rightarrow 0,
\end{align*}%
which results in lossless information compression. Finally, we have the
total loss function $\mathcal{L}(\theta _{en},\theta _{de})$ given by%
\begin{equation}
\mathcal{L}(\theta _{en},\theta _{de})=\sum_{i=1}^{3}\omega _{i}\mathcal{L}%
_{i}(\theta _{en},\theta _{de}),  \label{eq3-21}
\end{equation}%
where $0<\omega _{i}\leq 1$ denotes a weighting factor. Fig. \ref{fig3}
sketches the learning procedure proposed above, in which

\begin{itemize}
\item the NNs, $\mathcal{R}\mathcal{N}\mathcal{N}_{en}\left( \theta _{en},%
\left[ 
\begin{array}{c}
u^{(i)} \\ 
y^{(i)}%
\end{array}%
\right] \right) ,\mathcal{R}\mathcal{N}\mathcal{N}_{de}\left( \theta
_{de},v\right) ,$ in the AE block are learnt using an optimization algorithm,

\item $\left[ 
\begin{array}{c}
\hat{u} \\ 
\hat{y}%
\end{array}%
\right] $ and $\hat{v}$ are the intermediate variables used for the
calculation of the regularized terms $\mathcal{L}_{2}(\theta _{de},\theta
_{en})$ and $\mathcal{L}_{3}(\theta _{de},\theta _{en}),$ which, together
with the latent variable $v,$ process data $\left[ 
\begin{array}{c}
u \\ 
y%
\end{array}%
\right] $ and $\mathcal{L}_{1}(\theta _{de},\theta _{en}),$ build the loss
function $\mathcal{L}(\theta _{de},\theta _{en})$ (\ref{eq3-21}) for the
optimization.
\end{itemize}

\noindent We name the above AE \textit{I}nner-\textit{A}utoencoder (I-AE).
It is worth to emphasize that the I-AE-based solution proposed above is a
pure data-driven solution. It requires (i) no process model, e.g. as given
by (\ref{eq3-10}), and (ii) no solution of differential equations like (\ref%
{eq3-14a})-(\ref{eq3-14c}). So far, it is a general and practical solution
for detecting faults in nonlinear dynamic systems.

\section{Data experimental study\label{sub5}}

In this section, the main results of our work on control theoretically
guided training of autoencoders for detecting faults in nonlinear dynamic
processes are evaluated on a three-tank system simulator. The objective of
the evaluation study is twofold:

\begin{itemize}
\item comparison of the I-AE trained under the loss function (\ref{eq3-21})
with a standard AE regarding their capability of performing fault detection,

\item explanation of the role of the latent variable in approaching an
optimal fault detection.
\end{itemize}

\subsection{Experimental setting and data description}

Three-tank systems (TTS) have typical characteristics of chemical processes
and are widely accepted as a benchmark process in research and application
domains of process control and fault diagnosis. The simulator adopted in our
work is the real-time laboratory setup TTS20 that is in operation in the AKS
labor since more than 20 years \citep{Ding2008,Ding2020}. As schematically
sketched in Fig. \ref{fig_tts}, TTS20 consists of three water tanks that are
connected through pipes. Water from a reservoir is pumped into tank 1 and
tank 2, respectively. TTS20 is a nonlinear dynamic control system. Two input
variables, $u_{1}$ and $u_{2},$ manipulate the incoming mass flow into tank
1 and tank 2. In order to regulate the water levels in tank 1 and tank 2,
two PI-controllers are used with the output variables $y_{1}$ and $y_{2},$
which are measured by two level sensors. The reader is referred to \cite%
{Ding2014,Ding2020} for a detailed description of the nonlinear dynamic
model, the asset and controller parameters. The simulator has been developed
in the MATLAB/SIMULINK software environment with a sampling time equal to $%
1s $, and simulated sensor and actuator noises. It demonstrates excellent
simulation results, and is widely used in teaching programs and research
projects.

\begin{figure}[t]
\centering\includegraphics[width=0.45\textwidth]{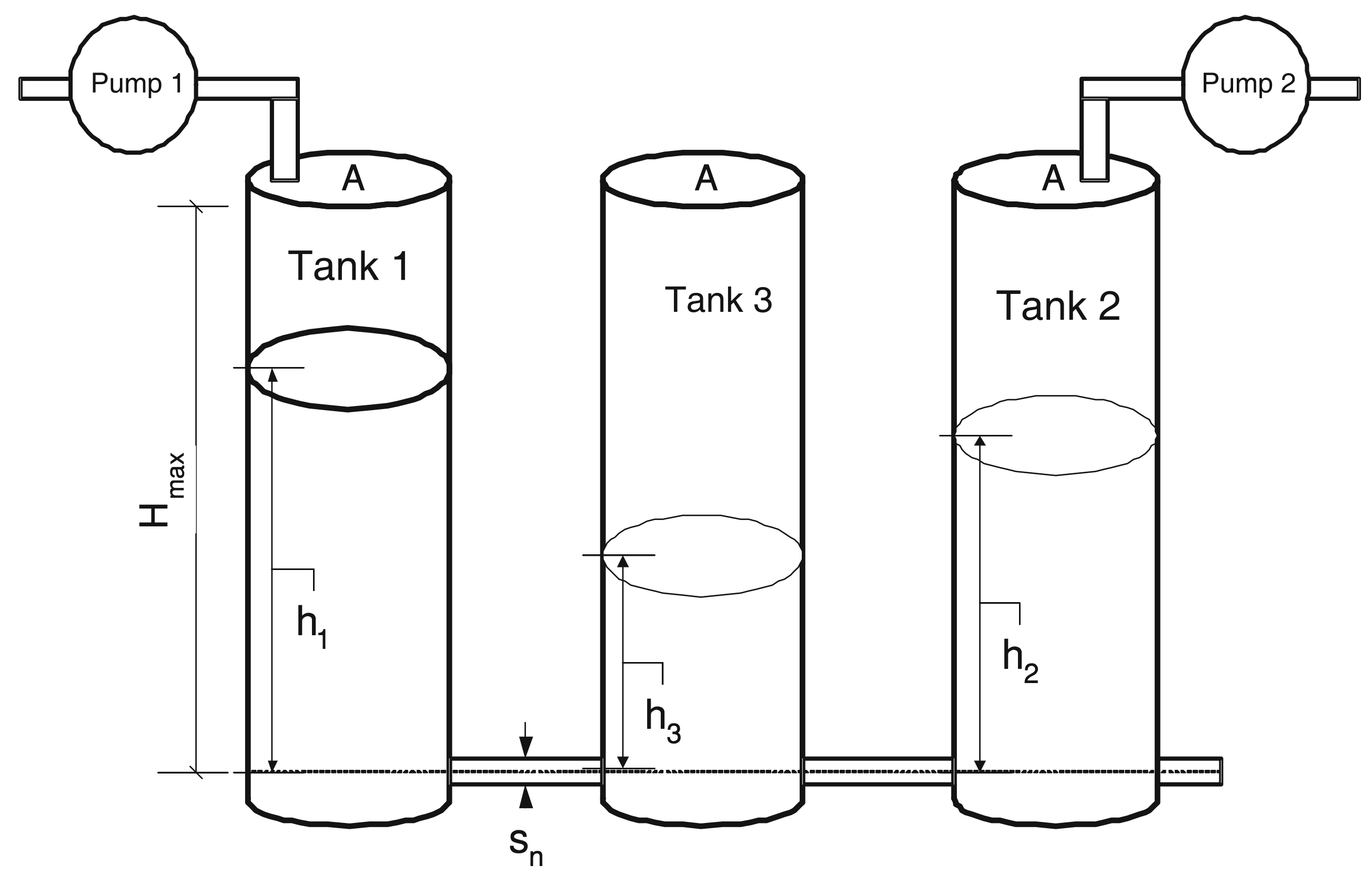}
\caption{Schematic of three-tank system TTS20}
\label{fig_tts}
\end{figure}

\subsection{Design of evaluation program and data description}

In order to guarantee that the NNs to be learnt fully model the nonlinear
system dynamics, the reference signals are randomly generated so that over $%
85\%$ of the operation region of the process are covered. An example is
given in Fig. \ref{fig_y1ref}, which shows the water level in tank 1, $y_{1}$%
, with reference signal $y_{1ref}$.

\begin{figure}[t]
\centering\includegraphics[width=0.45\textwidth]{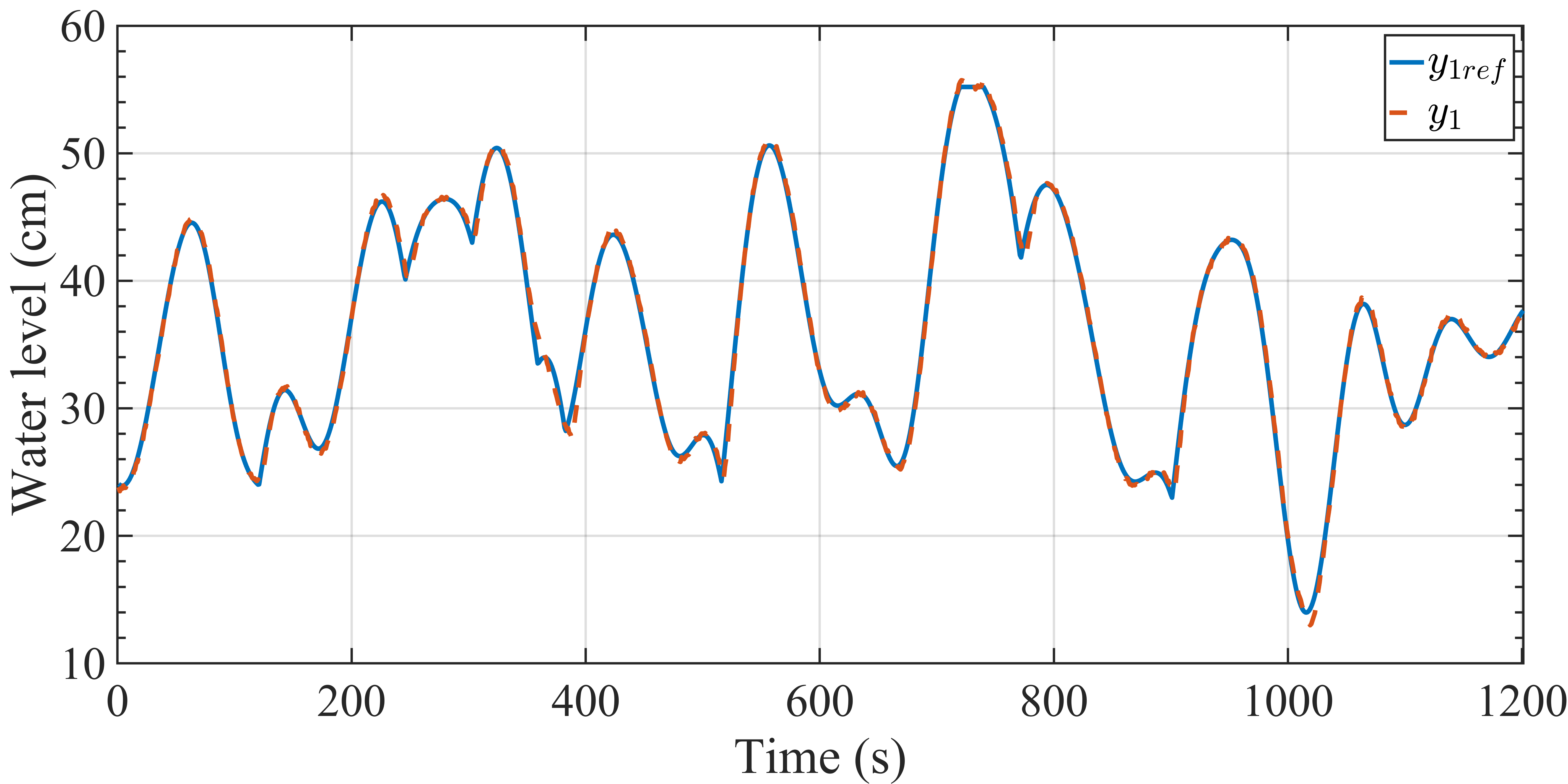}
\caption{Water level in tank 1}
\label{fig_y1ref}
\end{figure}

To prepare the training data, $1,200,000$ fault-free samples are
continuously generated by TTS20-simulator. They are randomly divided into a
training set and a validation set in a ratio of 7:3. The following three
types of faults are simulated:

\begin{itemize}
\item $10\%$ leakage fault in tank 1,

\item $5\%$ sensor gain fault in the water level sensor mounted on tank 2,

\item sensor gain fault in the water level sensor mounted on tank 2 with
stepwise changes from $5\%$ to $10\%$ and $15\%$.
\end{itemize}

\noindent To achieve a reliable fault detection, a moving window with the
window length equal to $100s$ is adopted for the evaluation purpose.

\subsection{Configuration of autoencoders and training}

Recurrent neural networks (RNNs) are known to be capable of learning time
evolutionary features of dynamic systems. For our purpose, a typical RNN,
the long short-term memory (LSTM) NN \citep{lstm}, is applied for building
the autoencoder. In our experiments, the following hyperparameter
combination is adopted:

\begin{itemize}
\item four layers in the encoder and decoder, respectively,

\item two LSTM units are included in each layer.
\end{itemize}

\noindent With these hyperparameters, all models presented below are set and
trained.

In order to showcase the capability of the control theoretically guided
autoencoder learning, two autoencoders are trained by means of two different
loss functions. The autoencoder trained under the loss function $\mathcal{L}%
(\theta _{de},\theta _{en}),$ 
\begin{equation*}
\mathcal{L}(\theta _{de},\theta _{en})=\mathcal{L}_{1}(\theta _{de},\theta
_{en})=\frac{1}{M}\sum_{i=1}^{M}\left\Vert \left[ 
\begin{array}{c}
u^{(i)} \\ 
y^{(i)}%
\end{array}%
\right] -\left[ 
\begin{array}{c}
\hat{u}^{(i)} \\ 
\hat{y}^{(i)}%
\end{array}%
\right] \right\Vert _{2}^{2},
\end{equation*}%
is denoted by AE with $\mathcal{L}_{1},$ while the I-AE trained by the loss
function $\mathcal{L}(\theta _{de},\theta _{en})$ given in (\ref{eq3-21}),
i.e.%
\begin{equation*}
\mathcal{L}(\theta _{en},\theta _{de})=\sum_{i=1}^{3}\mathcal{L}_{i}(\theta
_{en},\theta _{de}).
\end{equation*}%
In Fig. \ref{fig_loss1}, the values of the loss functions of both AEs during
the training process are shown for comparison. It is apparent that

\begin{itemize}
\item in both cases, the training process converges properly, and

\item the AE with $\mathcal{L}_{1}$ is learnt faster with a smaller
reconstruction error than I-AE.
\end{itemize}

\noindent Moreover, according to the final reconstruction error and on the
demand for an FAR upper bound equal to $0.05$, threshold is set to be $%
0.2201 $. 
\begin{figure}[t]
\centering
\includegraphics[width=0.48\textwidth]{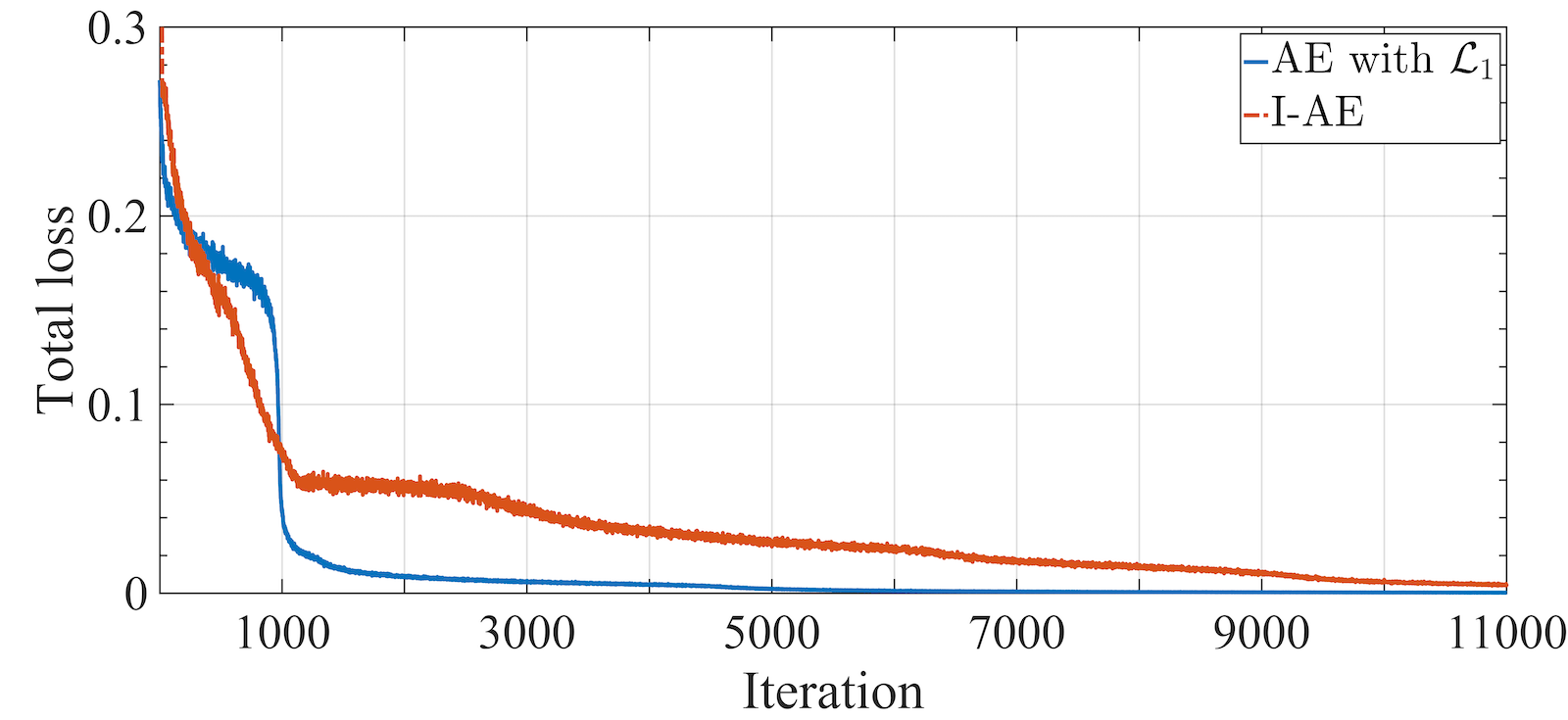}
\caption{Training loss functions in comparison}
\label{fig_loss1}
\end{figure}

\subsection{Fault detection performance evaluation}

The trained AE-based fault detection systems, AE with $\mathcal{L}_{1}$ and
I-AE, are tested aiming at examining and comparing their fault detection
performance. To this end, the following tests are designed:

\begin{itemize}
\item the TTS20-simulator runs for $400s$ under fault-free operation
conditions, and a fault is injected at the time instant $401s$ with a
duration of $400s,$

\item the AE-based fault detection system runs for $800s,$ and

\item the above procedure is repeated $1000$ times under random operation
conditions.
\end{itemize}
\begin{figure}[t]
\centering
\subfigure[Detection result of AE with
$\mathcal{L}_1$]{\includegraphics[width=0.47\textwidth]{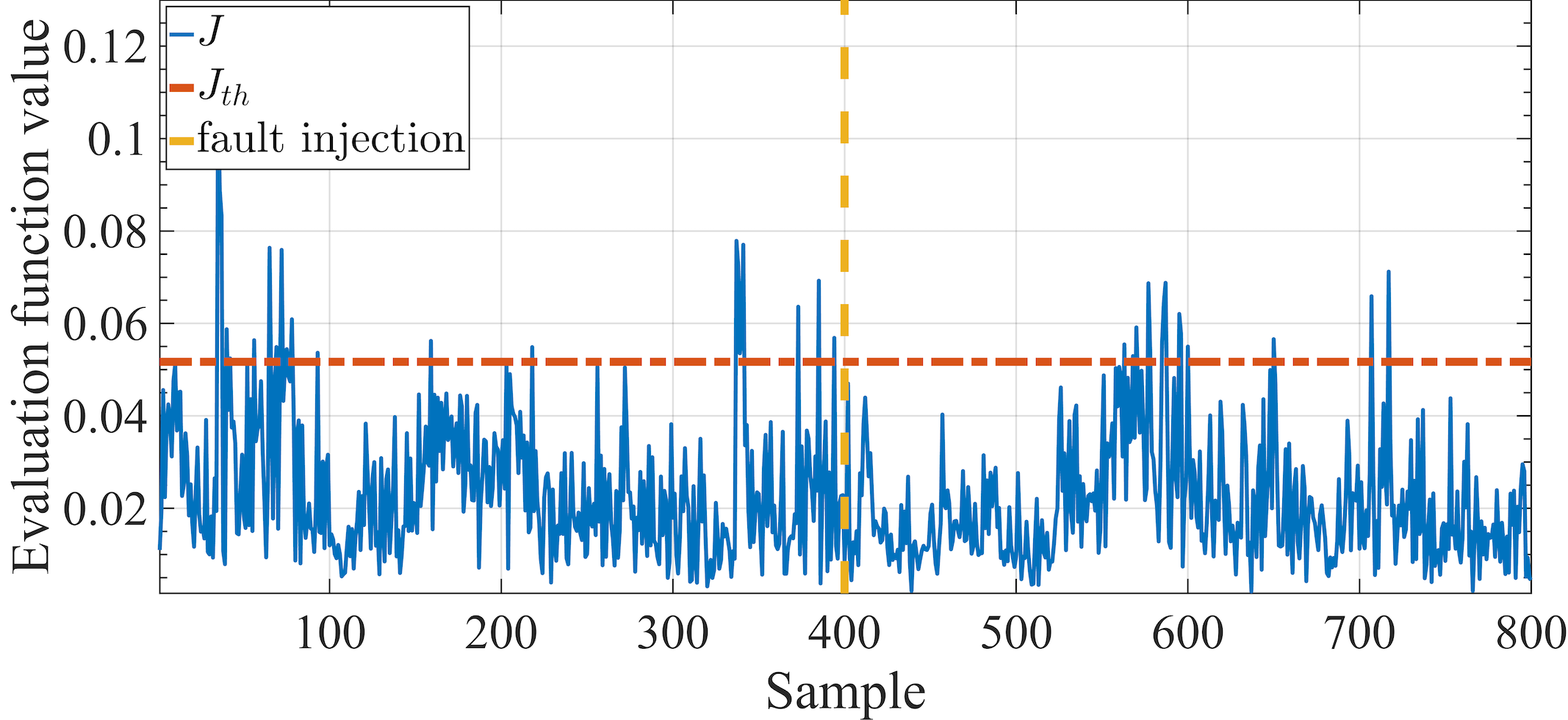}}%
\hspace{0.4in} 
\subfigure[Detection result of
I-AE]{\includegraphics[width=0.45\textwidth]{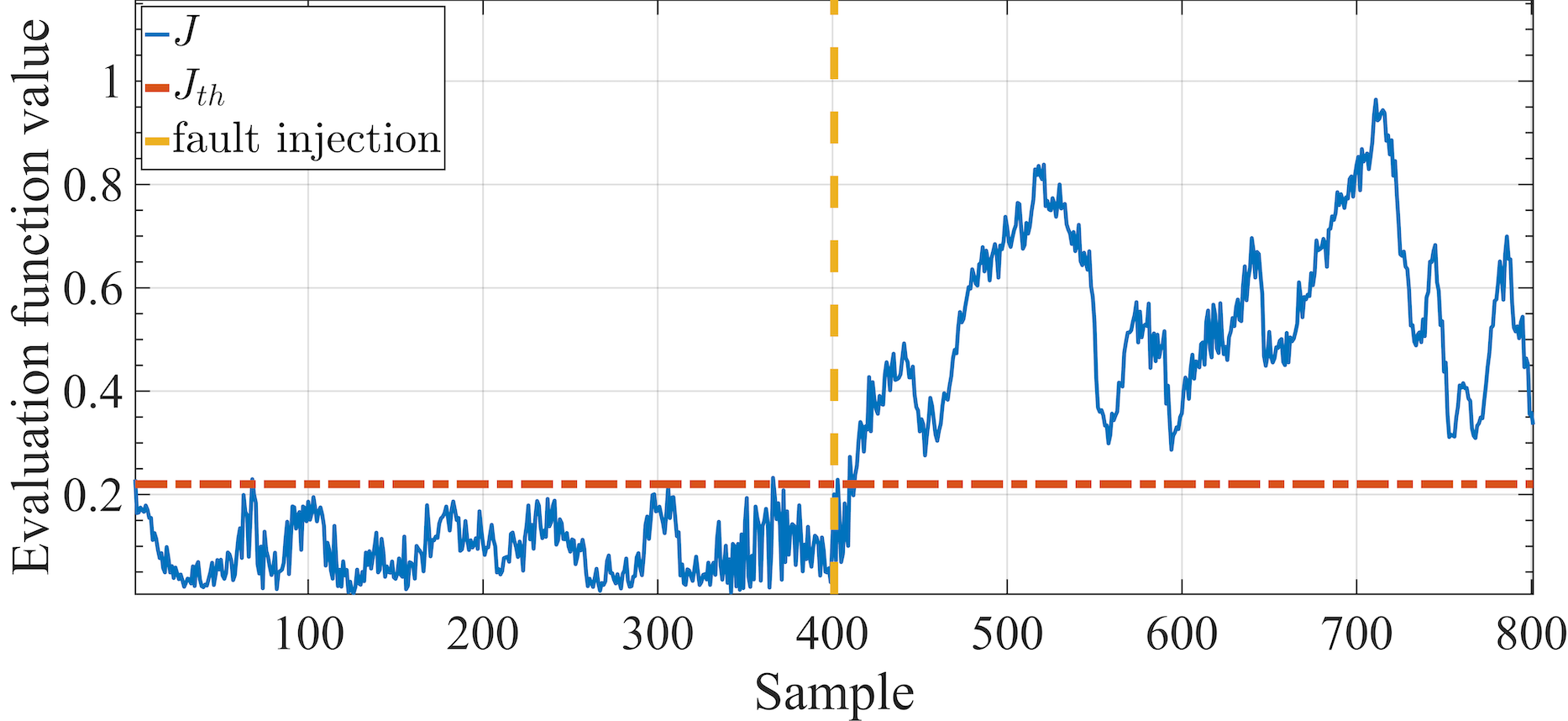}}
\caption{Detection of $10\%$ leakage fault in tank 1}
\label{fig_fd1}
\end{figure}

\begin{figure}[t]
\centering
\subfigure[Detection result of AE with
$\mathcal{L}_1$]{\includegraphics[width=0.45\textwidth]{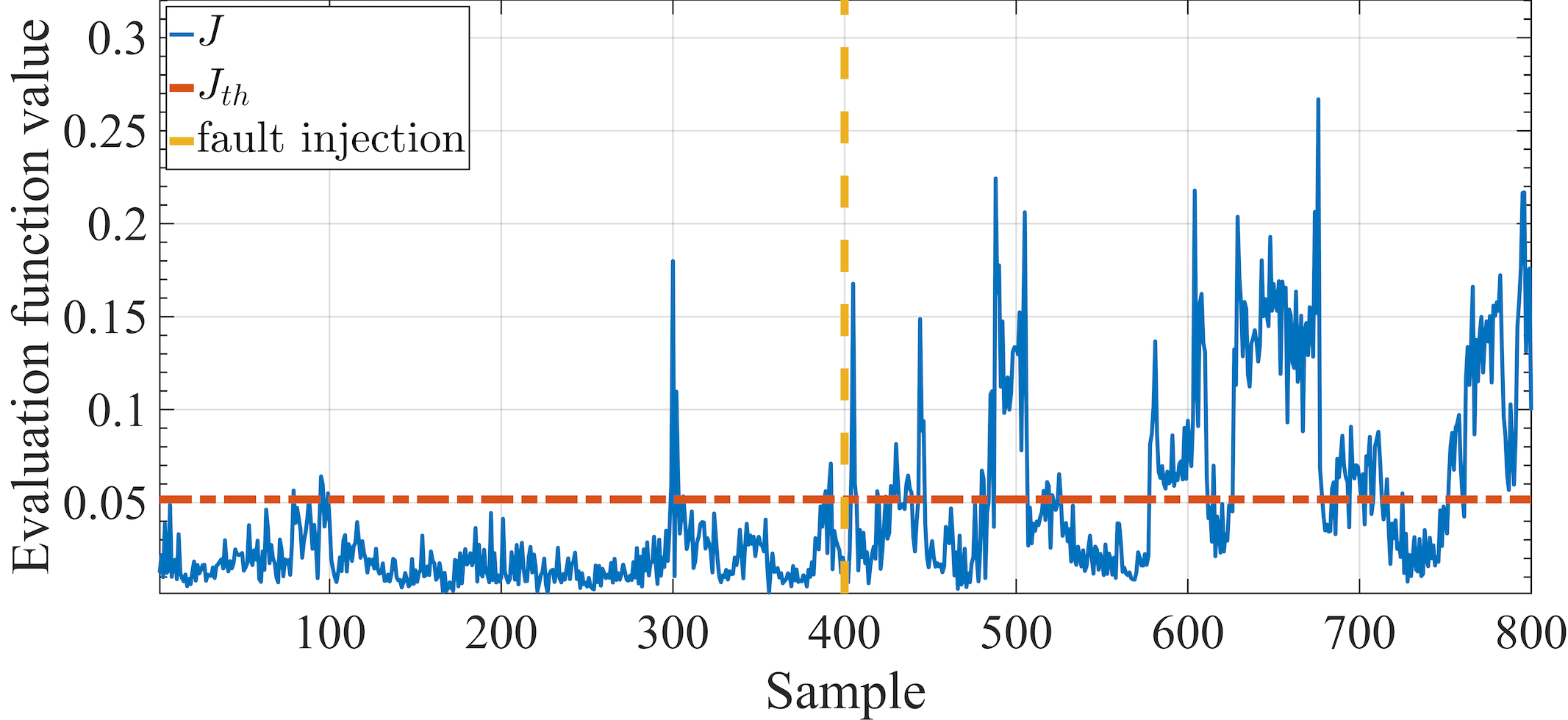}}%
\hspace{0.4in} 
\subfigure[Detection result of
I-AE]{\includegraphics[width=0.45\textwidth]{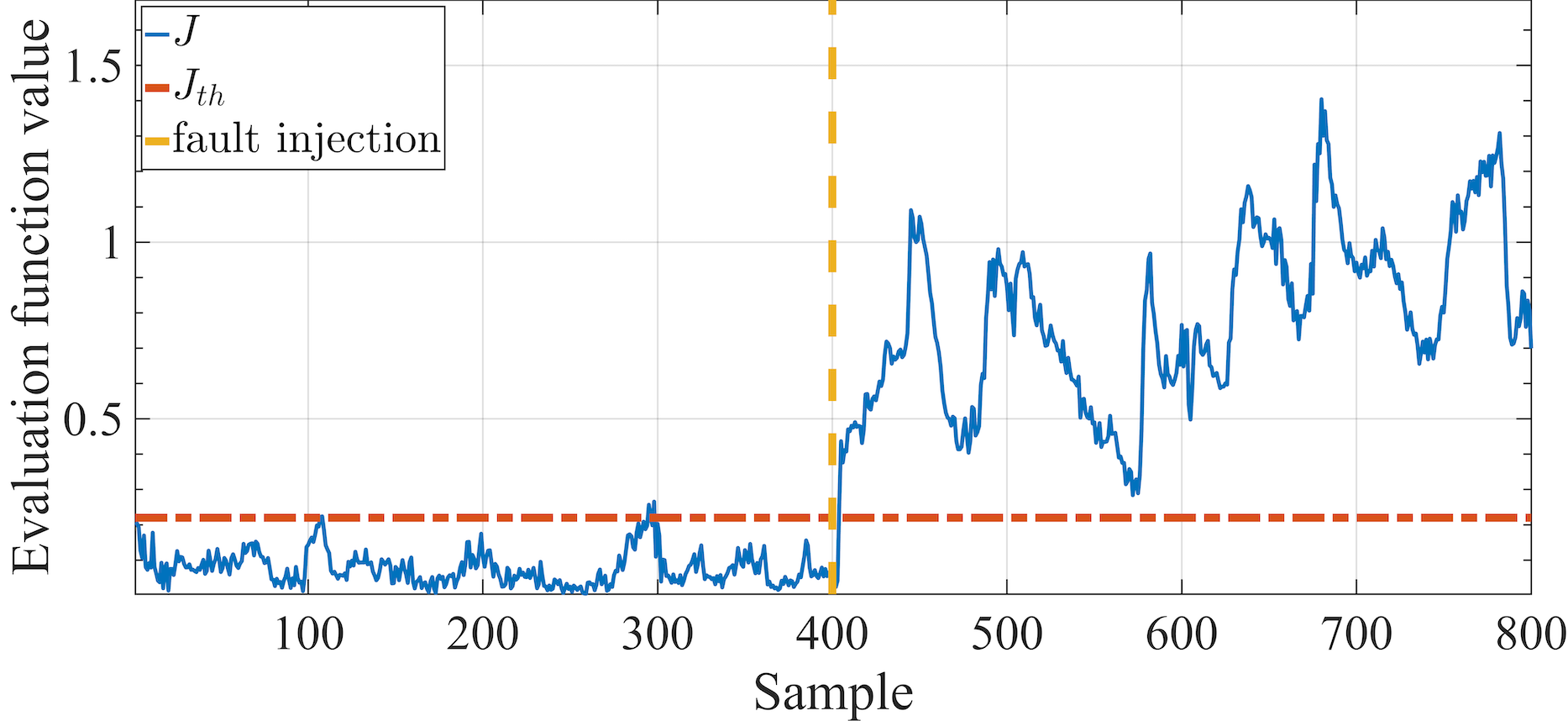}}
\caption{Detection of $5\%$ sensor gain fault in the water level sensor
mounted on tank 2}
\label{fig_fd2}
\end{figure}
\noindent Fig. \ref{fig_fd1} and Fig. \ref{fig_fd2} are examples of fault
detection results. Although these examples show a better detection
performance of I-AE than the one of the conventional AE with $\mathcal{L}%
_{1},$ performance evaluation metric statistics based on the simulation data
are necessary for a reasonable and comparable performance assessment. To
this end, FAR and missed detection rate (MDR) as well as accuracy and
F1-score that are commonly adopted for anomaly detection \citep{fawcett2006}
are under consideration. They are estimated using the detection data as
follows: 
\begin{align}
\text{FAR}& =\frac{N_{\text{FA}}}{N_{\text{FF}}},\text{ MDR}=\frac{N_{\text{%
MD}}}{N_{\text{F}}}, \\
\text{Accuracy}& =\frac{N-N_{\text{FA}}-N_{\text{MD}}}{N},N=N_{\text{F}}+N_{%
\text{FF}}, \\
\text{ F1-score}& =\frac{N_{\text{F}}-N_{\text{MD}}}{N_{\text{F}}-N_{\text{MD%
}}+\frac{1}{2}(N_{\text{FA}}+N_{\text{MD}})},
\end{align}%
where $N_{\text{F}}$ and $N_{\text{FF}}$ are the total numbers of faulty and
fault-free operations, and $N_{\text{FA}}$ and $N_{\text{MD}}$ are the
numbers of false alarms and miss detections, respectively. Thus, $N-N_{\text{%
FA}}-N_{\text{MD}}$ is the total number of correct decisions, $N_{\text{F}%
}-N_{\text{MD}}$ is the number of correctly detected faults, and $N$ is the
total number of the operations.

\begin{Rem}
Let faulty and fault-free operations be positive and negative classes,
respectively. In the framework of anomaly detection \citep{fawcett2006}, $N_{%
\text{F}}$ and $N_{\text{FF}}$ correspond to the total numbers of positive
class and negative class, and $N_{\text{FA}}$ and $N_{\text{MD}}$ are the
numbers of false positives and false negatives, respectively.
\end{Rem}

In Table \ref{tab_fd}, the performance evaluation results are listed for
both AEs. It is obvious that

\begin{itemize}
\item both AEs are at the similar FAR level about $0.05,$ as required,

\item the I-AE is of a considerably lower MDR than the conventional AE with $%
\mathcal{L}_{1},$ and consequently,

\item accuracy and F1-score of the I-AE are remarkably higher than the ones
of the conventional AE with $\mathcal{L}_{1}.$
\end{itemize}

\begin{table}[t]
\caption{Fault detection performance evaluation}
\label{tab_fd}\vspace{0.1in} \center
\renewcommand\arraystretch{1.3} 
\begin{tabular}{m{40pt}m{60pt}m{60pt}}
\hline
& \textbf{AE with $\mathcal{L}_1$} & \textbf{I-AE} \\ \hline
FAR & 0.051$\pm$0.004 & 0.049$\pm$0.007 \\ 
MDR & 0.411$\pm$0.039 & 0.034$\pm$0.002 \\ 
Accuracy & 0.537$\pm$0.038 & 0.915$\pm$0.011 \\ 
F1-score & 0.267$\pm$0.106 & 0.916$\pm$0.013 \\ \hline
\end{tabular}%
\end{table}

\subsection{Test on the role of latent variable}

To illustrate the influence of the proposed regularized terms on the latent
variable in the I-AE and hence to gain a deeper insight into the role of the
latent variable, an ablation study is designed as follows: a third AE, named
AE with $\mathcal{L}_{1}\sim \mathcal{L}_{4},$ is trained under the loss
function\ 
\begin{equation*}
\setlength{\abovedisplayskip}{6pt}\setlength{\belowdisplayskip}{6pt}
\mathcal{L}(\theta _{en},\theta _{de})=\sum_{i=1}^{3}\mathcal{L}_{i}+0.001%
\mathcal{L}_{4},
\end{equation*}%
where $\mathcal{L}_{i},i=1,2,3,$ are the ones given in (\ref{eq3-41})-(\ref%
{eq3-43}), and $\mathcal{L}_{4}$ is an additional regularized term defined by
\begin{align}
& \mathcal{L}_{4}(\theta _{en})=\left\vert \left\Vert \left[ 
\begin{array}{c}
u \\ 
y%
\end{array}%
\right] \right\Vert _{2}-\left\Vert v\right\Vert _{2}\right\vert
\label{eq3-44} \\
& =\left\vert \left\Vert \left[ 
\begin{array}{c}
u \\ 
y%
\end{array}%
\right] \right\Vert _{2}-\left\Vert \mathcal{N}\mathcal{N}_{en}\left( \theta
_{en},\left[ 
\begin{array}{c}
u \\ 
y%
\end{array}%
\right] \right) \right\Vert _{2}\right\vert  \notag
\end{align}%
that minimizes $\left\vert\left\Vert \left[ 
\begin{array}{c}
u \\ 
y%
\end{array}%
\right] \right\Vert _{2}-\left\Vert v\right\Vert _{2}\right\vert$. With $%
\mathcal{L}_{4} $, the latent variable $v$ will tend to preserve all the
information in $u$ and $y$, including uncertainty and redundancy. In other
words, $v$ contains, under such a training condition, the information about
more than the nominal operations. In the information theoretical context,
the latent variable $v$ is no more minimal sufficient statistic. Due to the
large amplitude of $\mathcal{L}_{4}$, it is weighted by a factor $0.001$.
The training loss of AE with $\mathcal{L}_{1}\sim \mathcal{L}_{4}$ is shown
in Fig. \ref{fig_loss3}. 

\begin{figure}[!t]
\centering
\includegraphics[width=0.44\textwidth]{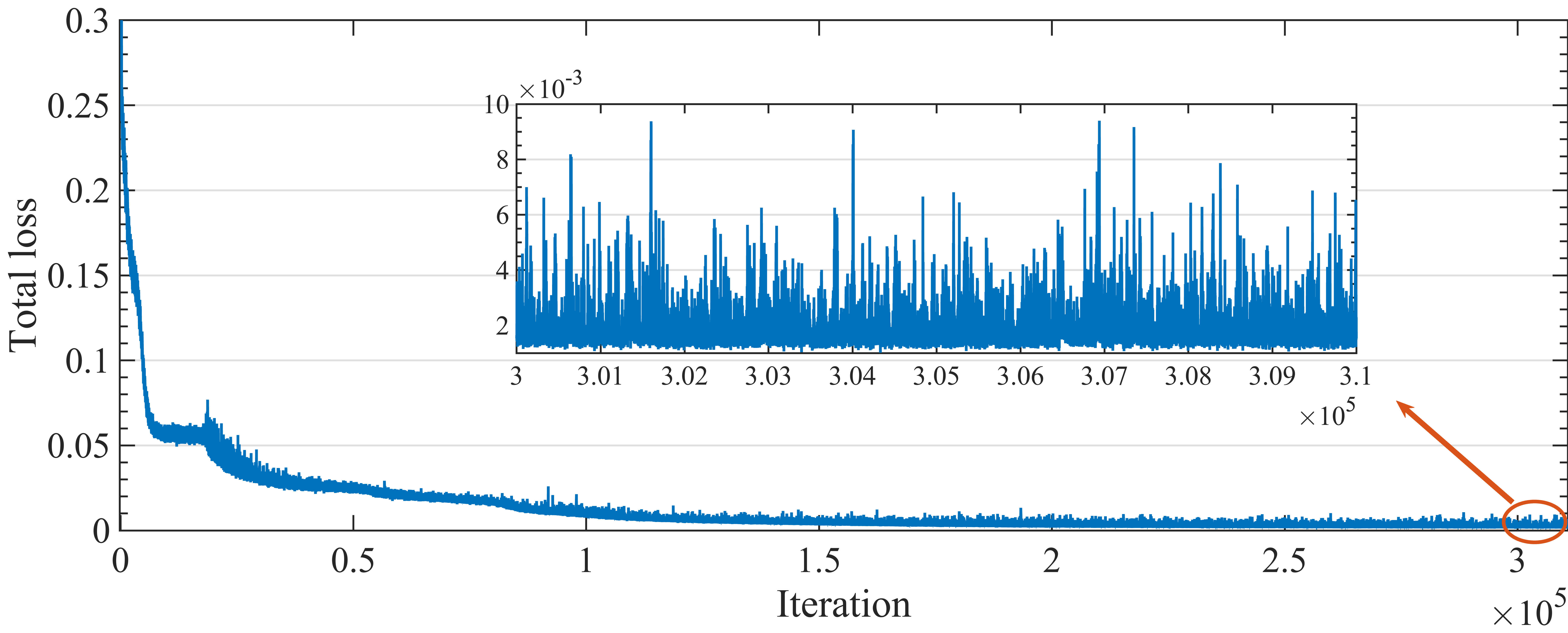}
\caption{Training loss of AE with $\mathcal{L}_1\sim\mathcal{L}_4$}
\label{fig_loss3}
\end{figure}
\begin{figure}[t]
\centering
\subfigure[Fault detection result of AE with
$\mathcal{L}_1$]{\includegraphics[width=0.45\textwidth]{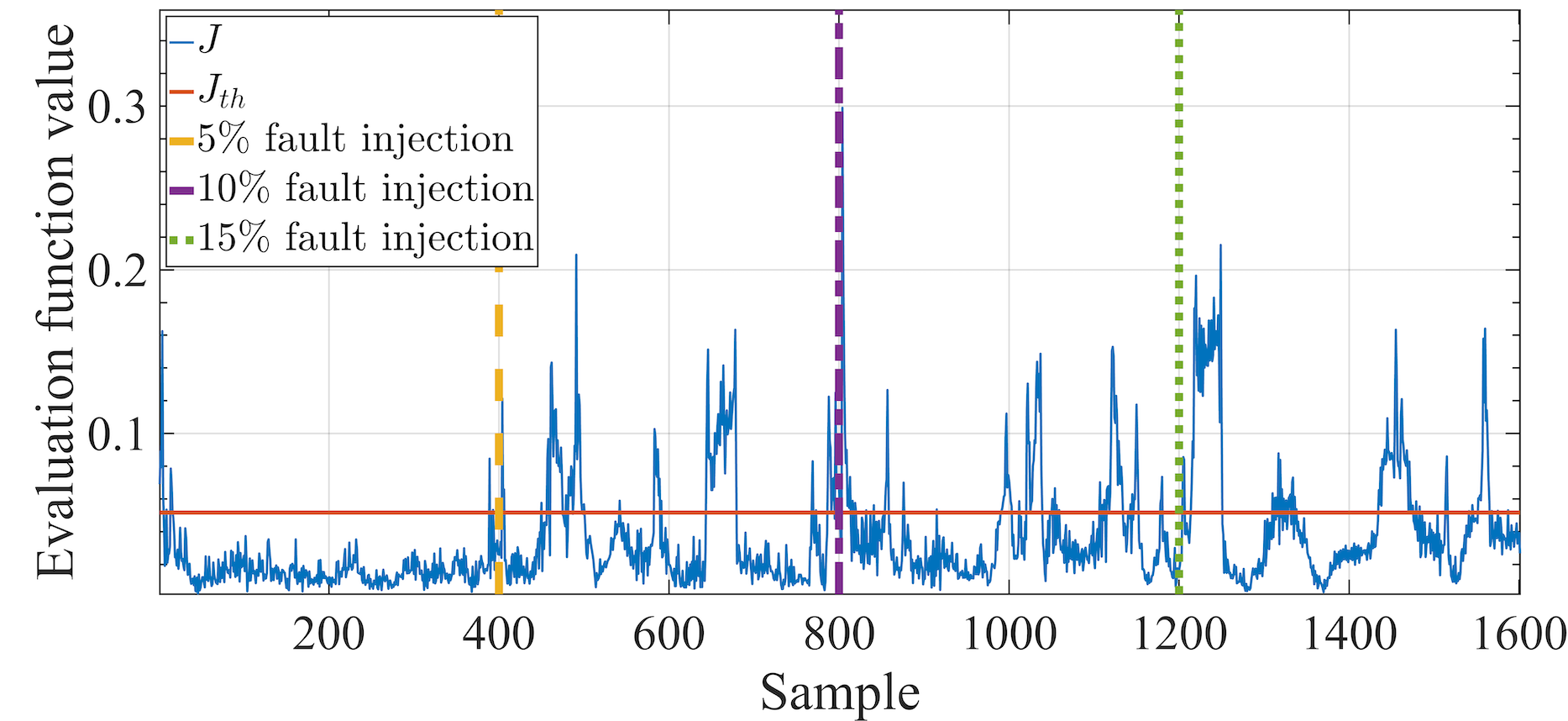}}%
\hspace{0.4in} 
\subfigure[Fault detection result of
I-AE]{\includegraphics[width=0.45\textwidth]{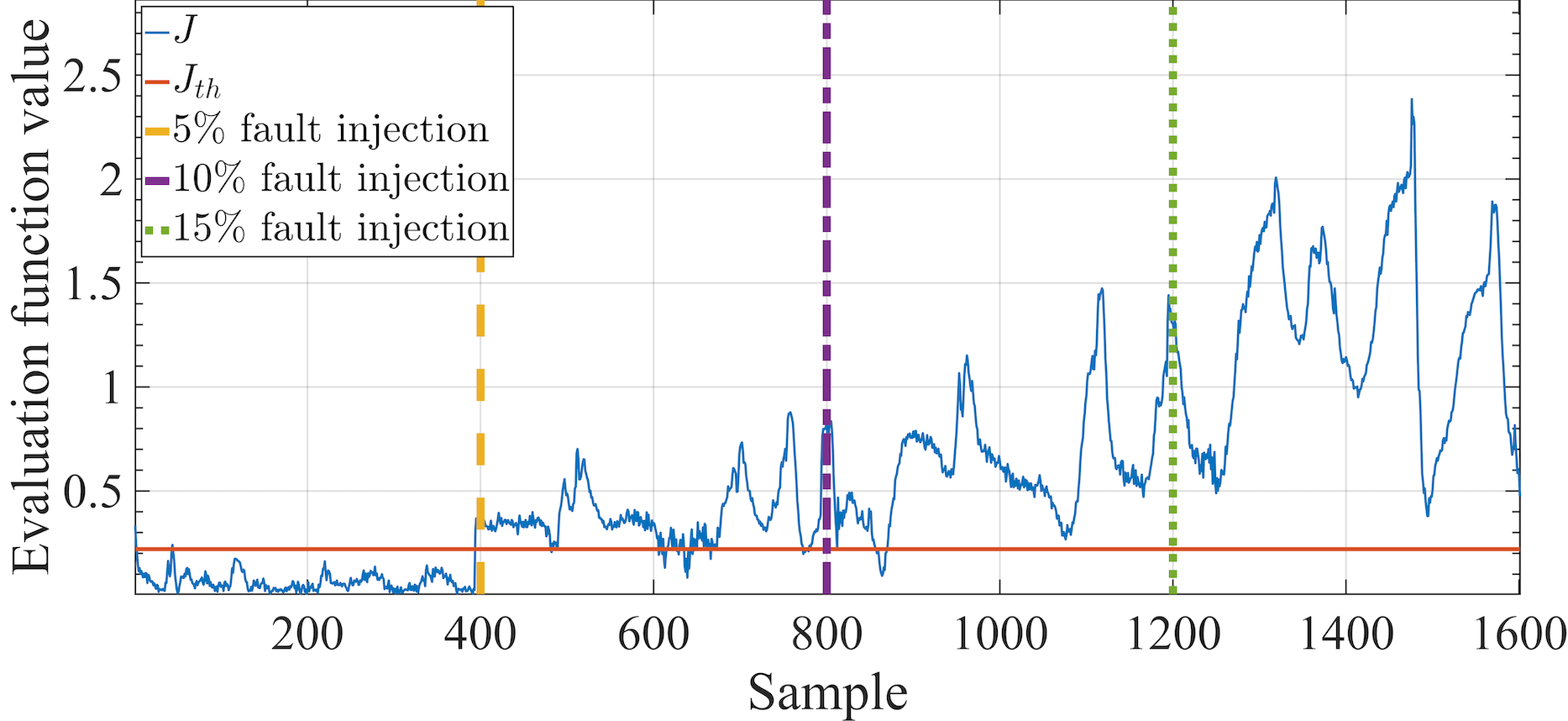}} 
\subfigure[Fault detection result of AE with
$\mathcal{L}_1\sim\mathcal{L}_4$]{\includegraphics[width=0.45\textwidth]{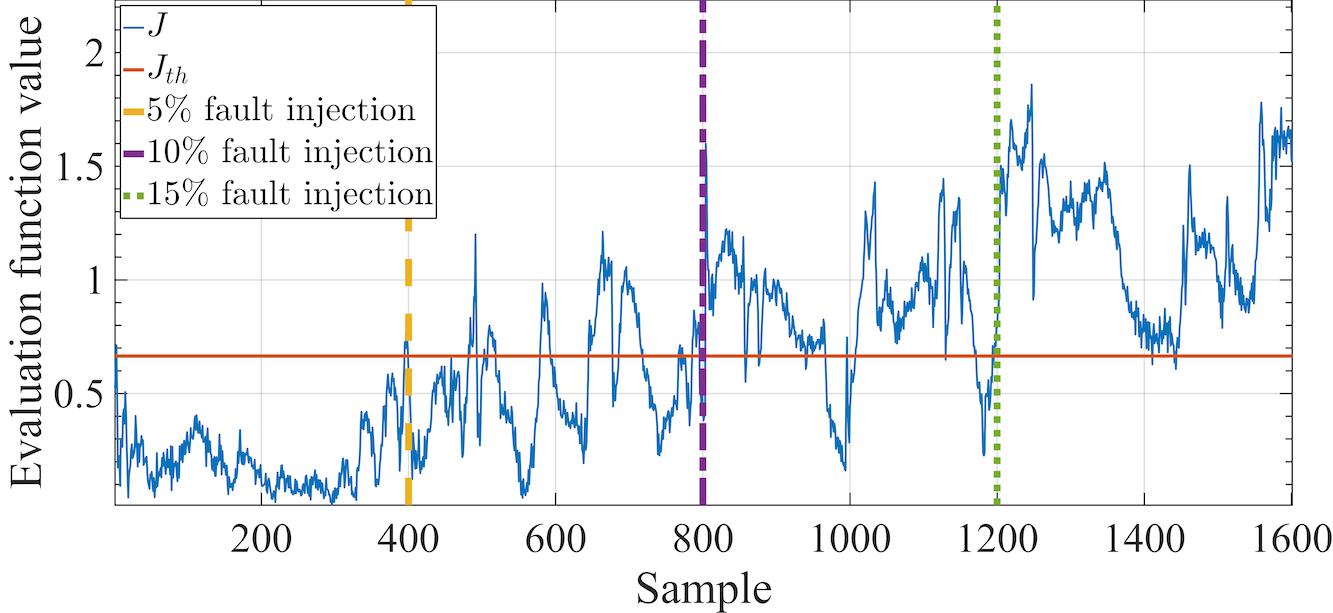}}
\caption{Detection of sensor gain fault in the water level sensor mounted on
tank 2 with stepwise changes from $5\%$ to $10\%$ and $15\%$}
\label{fig_res}
\end{figure}

The loss values of AE with $\mathcal{L}_{1}\sim \mathcal{L}_{4}$ are small
but not stable enough to fully converge, indicating that there is a
contradiction between $\mathcal{L}_{4}$ and the other loss terms. To assess
the fault detection performance and to compare with the other two AEs, the
same tests on detecting the sensor gain fault in the water level sensor
mounted on tank 2 with stepwise changes from $5\%$ to $10\%$ and $15\%$ are
conducted. A test example is shown in Fig. \ref{fig_res}, while the
detection performance evaluation results are summarized in Table \ref%
{tab_fd2}. It can be clearly seen that the introduction of $\mathcal{L}_{4}$
in the loss function results in an AE-based detection system with degraded
fault detectability in comparison with I-AE, the AE trained with a loss
function consisting of the regularized terms $\mathcal{L}_{2}$ and $\mathcal{%
L}_{3}.$ This result impressively demonstrates the role of a latent variable
as a minimal sufficient statistic.

\begin{table}[t]
\caption{Evaluation of fault detection performance}
\label{tab_fd2}\vspace{0.1in} \center
\renewcommand\arraystretch{1.3} 
\begin{tabular}{m{35pt}m{50pt}m{50pt}m{50pt}}
\hline
& \textbf{AE with $\mathcal{L}_1$} & \textbf{I-AE} & \textbf{AE with $%
\mathcal{L}_1$ $\sim\mathcal{L}_4$} \\ \hline
FAR & 0.050$\pm$0.016 & 0.051$\pm$0.012 & 0.052$\pm$0.021 \\ 
MDR & 0.355$\pm$0.031 & 0.036$\pm$0.014 & 0.091$\pm$0.022 \\ 
Accuracy & 0.593$\pm$0.039 & 0.913$\pm$0.020 & 0.856$\pm$0.037 \\ 
F1-score & 0.542$\pm$0.053 & 0.927$\pm$0.016 & 0.875$\pm$0.030 \\ \hline
\end{tabular}%
\end{table}

\subsection{Analysis and summary}

As a summary of the experimental results, the following conclusions can be
drawn:

\begin{itemize}
\item Although the training process of the standard AE shows a higher
convergence rate and lower loss function value in comparison with the
training process of the I-AE, the I-AE demonstrates a remarkably higher
fault detectability, while the requirement on FAR is satisfied.

\item These experimental results verify our theoretical results: the fault
detection system, consisting of (\ref{eq3-17}) as the estimator (residual
generator), (\ref{eq3-18}) as the residual evaluator and (\ref{eq3-19}) as
the threshold, delivers the solution to the optimal fault detection problem
for nonlinear dynamic systems, formulated as maximizing fault detectability
while satisfying the FAR requirement. In this context, the control theoretic
results serve as a meaningful guideline for training the I-AE.

\item The reason behind these impressive results is that the I-AE is an
inner system, which ensures lossless information compression via the latent
variable, and consequently leads to maximal fault detectability, and

\item this is achieved by means of introducing the regularized terms $%
\mathcal{L}_{2}$ and $\mathcal{L}_{3}$ into the loss function, which enables
to train the AE being inner and lossless information compression.

\item It is noteworthy that the role of the latent variable in approaching
the optimal fault detection solution is of indispensable importance. The
latent variable can be interpreted as a minimal sufficient statistic in the
information theoretic framework. Moreover,

\item the I-AE based optimal fault detection system is realized in the
data-driven fashion without process models and analytical solutions of
(partial) differential equations.
\end{itemize}

\noindent We would like to mention that in the experiments, possible impacts
of NN types and structures on AEs and fault detection systems have not been
in the focus of our investigation. It can be expected that further
improvement of fault detection performance could be achieved by targeted
selection of NN types and structures, in particular when complex nonlinear
dynamic systems are under consideration. Furthermore, in our study, the
threshold setting has been realized in a simple and straightforward way.
Also in this regard, improvement of fault detection performance can be
expected, for instance, applying the probabilistic threshold setting methods 
\citep{Xue2020}.

\section{Conclusions\label{sub6}}

In this paper, an AE-based solution of optimal fault detection in nonlinear
dynamic systems has been studied and validated. As often demanded in
engineering applications, optimal fault detection is hereby formulated as
maximal fault detectability subject to FAR requirement. Dynamic systems
considered in this work differ from those objects and processes addressed by
the existing AE-based fault (anomaly) detection and classification methods
generally in their complex dynamics and corrupted hybrid uncertainties. In
control theory, there exist, on the one hand, well-established methods to
deal with such issues. On the other hand, their application to system design
and optimization requires the existence of analytical process models and
solutions of (partial) differential equations. The basic idea of this work
is to fuse the learning capability of NNs and AEs and rich control theoretic
knowledge to approach optimal fault detection in nonlinear dynamic systems
in the data-driven and learning fashion.

In the first part of our work, the coprime factorization technique has been
applied to the establishment of a fault detection framework, including SIR
and image subspace for LTI models, orthogonal projection onto system image
subspace and optimal fault detection system design. In this regard, the
concepts of latent variable for LTI systems and lossless information
compression have been introduced. On the basis of Hamiltonian system theory,
these results and concepts have been extended to nonlinear dynamic systems.
It has been proven that a nonlinear SIR system results in lossless
information compression and delivers an optimal fault detection solution if
it is inner.

Guided by the aforementioned control theoretic results, an AE-based optimal
fault detection system has been realized and validated. The core of this
part of our work is the control theoretically explained training of the AE.
To be specific, two regularized terms have been added in the loss function,
which results in the so-called I-AE, an AE that is inner. It has been
validated that the latent variable in the I-AE is lossless information
compression. A comprehensive data experimental study on the laboratory
TTS20-simulator has impressively demonstrated that

\begin{itemize}
\item the experimental results conform with the theoretic results, namely,
the I-AE-based fault detection system increases fault detectability
considerably in comparison with a standard AE-based detection system, while
satisfying the FAR requirement,

\item by means of the I-AE, the optimal fault detection system is realized
by learning in the data-driven fashion without process models and analytical
solutions of (partial) differential equations,

\item learning the latent variable being lossless information compression is
of indispensable importance to approach the optimal fault detection solution.
\end{itemize}

\noindent In parallel to the aforementioned investigations, information
theoretic aspects of the proposed I-AE have been studied as well. Inspired
by the discussion on information bottleneck in the context of feature
learning, in particular the concept of minimal sufficient statistic as an
optimal latent variable, information relations between the process data and
latent variable have been analyzed by means of mutual information rate. It
has been proven that, on assumption of the defined system model setting, the
mutual information rate of process data and the latent variable is zero.
This result implies that (i) the latent variable is a minimal sufficient
statistic and contains no uncertainties corrupted in the collected process
data, and thus (ii) the maximal fault detectability can be achieved thanks
to the maximal sensitivity of the evaluation function to the faulty
operations.

Our near future work will be devoted to the application study on complex
nonlinear dynamic systems and the extension to fault isolation and
classification. In course of this work, also research efforts will be made
on (i) impact analysis of NN types and structures on fault diagnosis
performance, and (ii) threshold setting in the probabilistic framework.

\bigskip

\textbf{Acknowledgement}: The authors are grateful to Dr. D. Zhao for the intensive and valuable discussions.

\bibliographystyle{model5-names}
\bibliography{ieeepesp}

\end{document}